\documentclass[journal]{IEEEtran}
\usepackage{color}
\usepackage{verbatim}
\usepackage{amsfonts}
\usepackage{amssymb}
\usepackage{stfloats}
\usepackage{cite}
\usepackage{graphicx}
\usepackage{psfrag}
\usepackage{amsmath}
\usepackage{array}
\usepackage{epstopdf}
\usepackage{authblk}
\usepackage{graphicx} 
\usepackage{amsthm} 
\usepackage{lipsum}
\usepackage{verbatim} 
\usepackage{authblk}
\usepackage{mathtools}
\usepackage{cuted}
\usepackage[lined,boxed,ruled]{algorithm2e}
\usepackage{algpseudocode}
\usepackage{framed} 
\usepackage{soul} 
\usepackage{subcaption}

\newtheorem{Lemma}{Lemma}
\newtheorem{Theorem}{Theorem}

\newtheorem{Proposition}{Proposition}
\newtheorem{Remark}{Remark}

\newtheorem{Corollary}{Corollary}

\newcommand{\removelatexerror}{\let\@latex@error\@gobble}

\begin{document}

\title{Digital Over-the-Air Computation:\\ Achieving High Reliability via Bit-Slicing}

\author{Jiawei~Liu, Yi~Gong, and Kaibin~Huang 
\thanks{Jiawei Liu is with the Department of Electrical and Electronics Engineering, The University of Hong Kong, Hong Kong, and also with the Department of Electrical and Electronics Engineering, Southern University of Science and Technology, Shenzhen 518055, China (e-mail: liujw@eee.hku.hk).

Yi Gong is with the Department of Electrical and Electronics Engineering, Southern University of Science and Technology, Shenzhen 518055, China (e-mail: gongy@sustech.edu.cn).

Kaibin Huang is with the Department of Electrical and Electronics Engineering, The University of Hong Kong, Hong Kong (e-mail: huangkb@eee.hku.hk).

Corresponding authors: Kaibin Huang; Yi Gong.}}
 
\IEEEpeerreviewmaketitle
\maketitle

\begin{abstract}
    6G mobile networks aim to realize ubiquitous intelligence at the network edge via distributed learning, sensing, and data analytics.
    Their common operation is to aggregate high-dimensional data, which causes a communication bottleneck that cannot be resolved using traditional orthogonal multi-access schemes.
    A promising solution, called over-the-air computation (AirComp), exploits channels' waveform superposition property to enable simultaneous access, thereby overcoming the bottleneck.
    Nevertheless, its reliance on uncoded linear analog modulation exposes data to perturbation by noise and interference.
    Hence, the traditional analog AirComp falls short of meeting the high-reliability requirement for 6G.
    Overcoming the limitation of analog AirComp motivates this work, which focuses on developing a framework for digital AirComp. 
    The proposed framework features digital modulation of each data value, integrated with the bit-slicing technique to allocate its bits to multiple symbols, thereby increasing the AirComp reliability.
    To optimally detect the aggregated digital symbols, we derive the optimal maximum a posteriori detector that is shown to outperform the traditional maximum likelihood detector.
    Furthermore, a comparative performance analysis of digital AirComp with respect to its analog counterpart with repetition coding is conducted to quantify the practical signal-to-noise ratio (SNR) regime favoring the proposed scheme.
    On the other hand, digital AirComp is enhanced by further development to feature awareness of heterogeneous bit importance levels and its exploitation in channel adaptation.
    Lastly, simulation results demonstrate the achivability of substantial reliability improvement of digital AirComp over its analog counterpart given the same channel uses.
\end{abstract}


\section{Introduction}
A key mission of 6G mobile networks is to realize ubiquitous intelligence at the network edge via distributed learning, sensing, and data analytics\cite{tong_6g_2021,peltonen_6g_2020}.
This will provide a platform for deploying a broad range of next-generation Internet-of-Things (IoT) applications such as healthcare, virtual reality, and industrial automation\cite{tong_6g_2021,zhou_edge_2019}.
A common network operation shared by the said use cases is for a server (or fusion center) to aggregate high-dimensional data uploaded by many mobile devices. 
Such data are in the form of high dimensional model updates, features of sensing data, and local computation results in the context of federated learning, distributed sensing, and distributed data analytics, respectively\cite{lim_federated_2020,xie_networked_2023,chen_--fly_2023}.
The resultant communication bottleneck cannot be resolved using traditional multi-access schemes as their attempts to orthogonalize multi-user data streams hinder their scalability.
Recently, addressing this issue has motivated researchers to explore a class of scalable, simultaneous access schemes, called over-the-air computation (AirComp)\cite{zhu_over_2021}. 
They exploit the property of waveform superposition to realize over-the-air aggregation of simultaneous data streams. 
However, due to the reliance on uncoded linear analog modulation, AirComp is exposed to perturbation by channel noise and interference and falls short of meeting the high-reliability requirement for 6G.
Overcoming the limitation motivates this work on developing the framework of digital AirComp to attain the desired high reliability.

AirComp was first explored in\cite{nazer_harnessing_2011} and\cite{abari_over--air_2016} for distributed sensing and further developed in\cite{zhu_mimo_2019} to enable spatial multiplexing targeting multimodel sensing.
Presently, the area of AirComp is experiencing fast growth as relevant techniques provide promising solutions for 6G edge intelligence\cite{zhu_over_2021,chen_guest_2022}.
Such techniques are especially popular in communication efficient federated learning, forming an area in its own right called over-the-air federated learning (AirFL)\cite{mohammadi_machine_2020,zhu_mimo_2019,yang_federated_2020,cao_optimized_2022,zhang_coded_2023,zhang_turning_2022,sery_over-federated_2021,wang_2023_spectrum}.
In this context, the purpose of AirComp is to aggregate local models or stochastic gradients uploaded by devices. The aggregation results are then applied to update a global model under training at a server\cite{mohammadi_machine_2020,zhu_mimo_2019,yang_federated_2020}.
Researchers have further developed advanced AirFL techniques including power control\cite{cao_optimized_2022,zhang_coded_2023,zhang_turning_2022}, adaptive precoding\cite{sery_over-federated_2021}, scheduling\cite{yang_federated_2020}, and interference suppression\cite{wang_2023_spectrum}.
Besides AirFL, progress has been made also in other directions where AirComp helps to overcome communication bottlenecks confronting distributed inference\cite{yilmaz_ensemble_2022}, integrated sensing and intelligence\cite{liu_over--air_2023}, distributed consensus\cite{lin_distributed_2023}, and distributed data analytics\cite{chen_--fly_2023}.
The aforementioned prior work all concerns analog AirComp. The reason is that linear analog modulation makes waveform superposition equivalent to the summation of pre-modulation data values. The desired equivalence may not hold if digital modulation-and-coding, often a nonlinear process, is applied. This gives rise to two drawbacks of AirComp.
One is the lack of protection by coding exposes computation to the channel hostility, making it difficult for the technology to attain high reliability.
The other is the incompatibility with existing mobile systems that are prevalently digital.

Addressing these issues has led to several initial studies on digital AirComp designs.
The earliest attempts involve the implementation of signSGD, one popular federated learning algorithm, such that the coefficients of stochastic gradient uploaded by each device are quantized into single bits while the aggregation operation helps to suppress the quantization errors via averaging\cite{zhu_one-bit_2021,jiang_cluster-based_2020}.
More recently, AirComp has been shown to be compatible with existing digital channel coding and generalized digital modulation\cite{you_broadband_2022,razavikia_channelcomp_2023,razavikia_sumcomp_2023,zhang_coded_2023}. Their basic principle is to treat AirComp as a special type of channel, encoding operations, or mapping between finite fields such that the designed computation/aggregation results can be recovered using a maximum-likelihood (ML) decoder or detector.
It is worth mentioning that nested lattice coding as studied in\cite{zhang_coded_2023} for AirFL is a particularly promising solution for digital AirComp as its inherent linearity ensures decodability of coded data after aggregation.
While prior work demonstrates the feasibility of digital AirComp, the field is still at its nascent stage.
In particular, a customized communication theory is yet to be developed and the optimal coding and modulation schemes for AirComp remain largely unknown.
In this work, we attempt to address the following three open issues.
\begin{enumerate}
    \item \textbf{Reliable multi-symbol AirComp}: The state-of-the-art digital AirComp only supports the mapping of each analog value to a \emph{single} digital symbol (e.g. QAM symbol), lacking the flexibility of traditional digital communication\cite{zhu_one-bit_2021,jiang_cluster-based_2020,you_broadband_2022,razavikia_channelcomp_2023,razavikia_sumcomp_2023,zhang_coded_2023}. For such designs termed single-symbol AirComp, their quantization errors are significant as each digital symbol comprises a relatively small number of bits (e.g., 4-8 bits). The techniques adopted therein for coding, modulation and detection are borrowed from traditional communication literature without considering computing integration. These lead to the underperformance of uncoded digital AirComp as opposed to its analog counterpart\cite{razavikia_channelcomp_2023,razavikia_sumcomp_2023,zhang_coded_2023}. The disadvantage calls for the finding of a flexible way of adding redundancy to digital AirComp, i.e., the development of \emph{multi-symbol AirComp}. Thereby, the desired ultra-high reliability for 6G can be achieved.
    \item \textbf{Optimal detection}: The optimality of the traditional ML detector, which is adopted in the current digital AirComp literature (see, e.g. \cite{razavikia_channelcomp_2023, razavikia_sumcomp_2023}), hinges on the assumption of equiprobable symbols. However, the aggregation operation in AirComp results in \emph{non-uniform} distribution of the symbols in a constellation. While the ML detector is no longer optimal, the optimal design remains unknown.
    \item \textbf{Bit-importance aware modulation and coding}: Different quantization bits (e.g., most significant versus least significant bits) have different importance levels in terms of their influences on AirComp error and thus warrant unequal protection during transmission. This computing-relevant issue is largely considered out-of-scope in the literature of traditional wireless communication techniques as they have been designed based on the communication-computing separation approach\cite{chen_guest_2022}. On the contrary, AirComp is inherently a communication-computing integration technology, making bit-importance-aware modulation and coding an important aspect of digital AirComp.
\end{enumerate}

As a result of investigating these issues, the main contributions and findings of this work are summarized as follows.

\begin{enumerate}
    \item \textbf{Multi-symbol digital AirComp framework}: To address open issue 1), the proposed framework specifies a complete sequence of transmitter/receiver operations for realizing multi-symbol digital AirComp. Central to the framework is a \emph{bit-slicing} technique that divides the bit sequence representing an analog data value into segments, each of which is transmitted as a single modulated (digital) symbol. Then the bit segments received by the receiver are assembled and used to reconstruct the transmitted analog value. We present computation-error analyses of both multi-symbol digital AirComp and its analog counterpart, namely analog AirComp with repetition coding. The results reveal that the former can attain significantly smaller errors than the latter in the high-reliability regime corresponding to moderate-to-high SNRs. This confirms the proposed technology's advantage in supporting high-reliability communication and computing.
    
    \item \textbf{Optimal detection}: To address open issue 2), we investigate the optimal maximum a posterior (MAP) detector for the preceding framework. Essentially, the design involves analyzing the distribution of aggregated digital symbols as a combinatorial problem. This allows the decision boundaries of the MAP detector to be derived. In contrast with traditional equidistant boundaries, the optimal ones for AirComp provide larger detection regions for more probable constellation points so as to minimize AirComp errors.
    
    \item \textbf{Importance aware bit-slicing}: Open issue 3) is addressed by exploring heterogeneity in bit importance to optimize the bit allocation (i.e., constellation sizes) of different digital symbols as generated by bit-slicing under a constraint on the total number of bits for representing an analog value. Such a combinatorial problem is solved numerically. The resultant scheme of importance-aware bit-slicing is demonstrated in experiments to significantly outperform uniform bit-slicing.               
\end{enumerate}

The remainder of the paper is organized as follows: in Section \ref{sec:system_model}, we describe the system model. Section~\ref{sec:digital_aircomp_framework} presents the framework of multi-symbol digital AirComp. 
The detector in the framework is optimized in Section~\ref{sec:optimal_detection}.
A Comparative performance analysis of digital and analog AirComp is conducted in Section~\ref{sec:error_analysis}.
The bit-slicing scheme is extended in Section~\ref{sec:extension_adaptive} to feature bit-importance awareness and channel adaptation. Numerical results are presented in Section~\ref{sec:numerical_results}, followed by conclusions in Section~\ref{sec:conclusion}

\section{System Model}
\label{sec:system_model}
Consider a distributed edge network with $K$ edge devices and an edge server.  
All devices simultaneously transmit their data streams to perform AirComp. The system model for the proposed digital AirComp is illustrated in Fig.~\ref{fig:system_model} with assumptions and metrics described in the subsections.

\begin{figure*}
    \centering
    \includegraphics[width=0.9\linewidth]{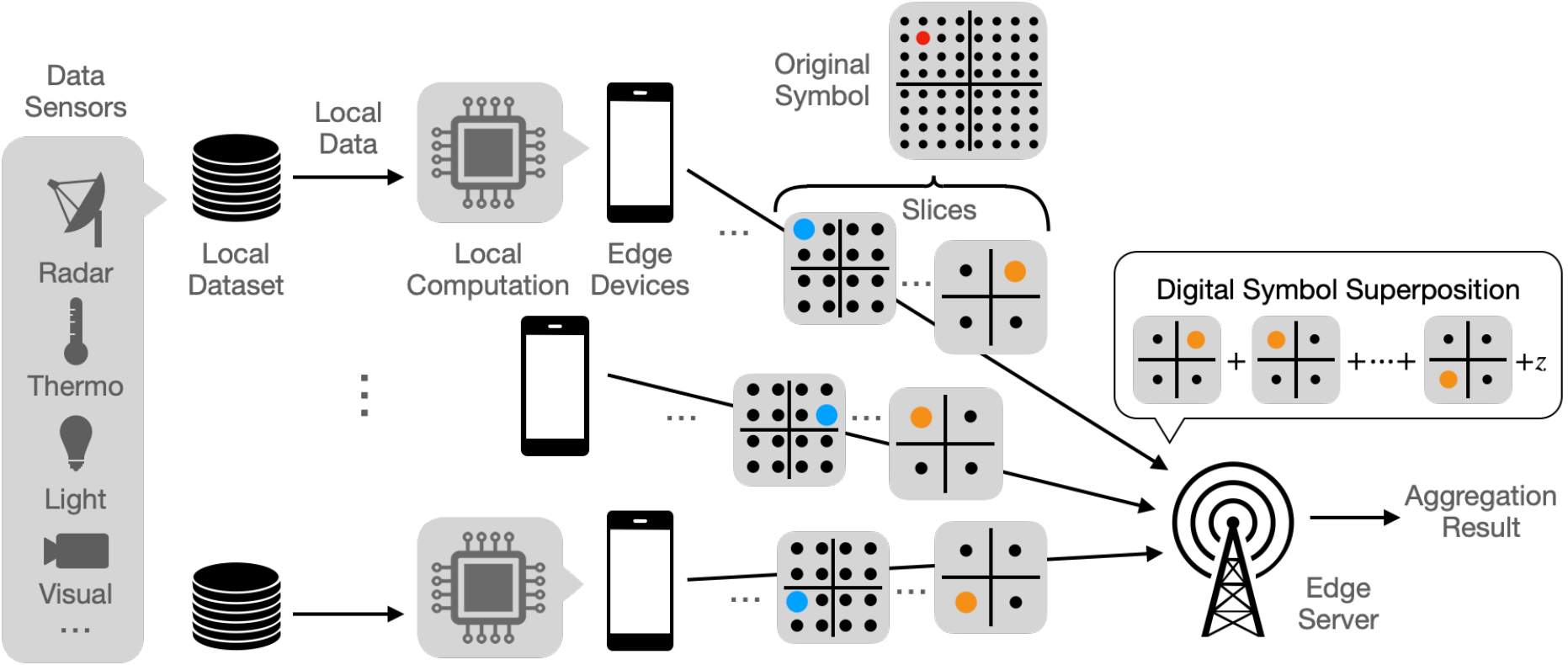}
    \caption{Overview of the digital AirComp system.}
    \label{fig:system_model}
\end{figure*}

\subsection{Computation Model}
\subsubsection{High-Precision Analog Source}
Given the prevalence of digital processors, sensing data are mostly stored and processed in digital format, and shared over a digital communication system.
In this context, an analog source is not truly analog (e.g., on an analog medium such as magnetic tapes), but instead, a high-precision digital source (with e.g., 32-bit or 64-bit precision) such that its difference from the ground-truth real-world value is negligible.
Then we can model the data vector each device intends to upload, say $\mathbf{x}_k$ for device $k$, as a \emph{continuous} random vector.
Let $\mathbf{x}_k$ consist of $M$ real elements: $\mathbf{x}_k=[x_k[1],x_k[2],...,x_k[M]]$.
It is assumed that each data vector be whitened and equalized locally such that its elements are identical and independently distributed (i.i.d) as uniform random variables: $x_k[m]\sim\mathcal{U}(x_{\min},x_{\max})$ for all $(k,m)$, where the constant $x_{\min}$ and $x_{\max}$ define their dynamic range.
The operations of whitening and equalization help to remove data redundancy and improve the efficiency of digital operations with finite dynamic ranges (e.g., analog-to-digital/digital-to-analog conversion\cite{dunn_1995_efficient}).

\subsubsection{Computation Task}
We consider the basic aggregation task underpinning typical computation algorithms for distributed sensing, learning, and inference.
Specifically, the task requires the server to estimate the summation of local data vectors, $\mathbf{y}=\sum_{k=1}^{K}\mathbf{x}_k$, from received signals due to transmission of $\{\mathbf{x}_k\}$ by devices.
A range of more sophisticated tasks, known as nomographic functional computation, can be implemented by adding to aggregation suitable pre-/post-processing\cite{zhu_mimo_2019}.
Examples include averaging, geometric mean, multiplication, weighted sum, and even maximization\cite{zhu_mimo_2019,liu_over--air_2023}.

\subsection{Communication Model}
\label{subsec:communicationmodel}
\subsubsection{Channel Model}
Consider the multi-access channel in Fig. \ref{fig:system_model}. Time is divided into slots, each spanning a single symbol duration. Consider an arbitrary device, say the $k$-th device. The channels between the devices and the server are assumed to be block-fading, where the channel gain $h_k$ remains unchanged in a channel coherence duration consisting of $R$ symbol slots.
Given channel diversity (achieved by, e.g., coherent combining), the channel power gain is assumed to follow the Chi-square distribution, i.e., $|h_k|^2\sim\chi^2(\kappa)$, where $\kappa\geq1$ denotes the degrees of freedom.

\subsubsection{AirComp}
All edge devices simultaneously upload their data vectors, $\{\mathbf{x}_k\}_{k=1}^{K}$, by transmit the corresponding $R\times1$ modulated symbol vectors, denoted as $\{\mathbf{m}_k\}_{k=1}^{K}$, over the multi-access channel.
Assuming synchronized carrier frequency and symbol timing, the received symbol vector at the server is given as:
\begin{equation}
    \label{eq:channel}
    \mathbf{r} = \sum_{k=1}^{K}h_k\sqrt{p_k}\mathbf{m}_k + \mathbf{z},
\end{equation}
where $p_k$ denotes the transmitting power of the $k$-th device, and $\mathbf{z}$ the additive white Gaussian noise (AWGN), with i.i.d $\mathcal{CN}(0,\sigma_z^2)$.

\subsubsection{Channel-Inversion Power Control}
Assuming that channel state information (CSI) is perfectly known at the server and devices through channel estimation and feedback. 
Considering \eqref{eq:channel}, to accomplish the aforementioned computation task, the power of each device, $p_k$, is controlled to invert the channel, $h_k$: $p_k=\rho^2/|h_k|^2$, where $\rho$ denotes a given scaling factor selected to ensure that the average transmit power constraint is satisfied for all devices\cite{zhu_broadband_2020}.
Note that with the Chi-square channel power distribution, the channel inversion requires finite average transmission power: $\mathbb{E}\left[\frac{\rho^2}{|h_k|^2}\right]=\rho^2/\kappa$. 
It follows that the receive signal-to-noise ratio (SNR) for an individual signal is $\gamma=\rho^2/\sigma_z^2$.

\subsection{Performance Metric}
\label{subsec:performance_metric}
Following from the literature (see, e.g.,\cite{zhu_mimo_2019}), the performance of AirComp is measured by the mean square error (MSE) between the ideal and estimated computation results, termed AirComp error and defined as:
\begin{equation}
    \label{eq:MSE}
    \mathcal{E}\triangleq\mathbb{E}\left[\left(\hat{y}-y\right)^2\right],
\end{equation}
By normalizing the AirComp error by the expected energy of the aggregation result, we obtain the normalized AirComp error metric, denoted as $\overline{\mathcal{E}}$:
\begin{equation}
    \label{eq:nMSE} \overline{\mathcal{E}}\triangleq\frac{\mathcal{E}}{\mathbb{E}\left[y^2\right]}.
\end{equation}

\section{Overview of Digital AirComp}
\label{sec:digital_aircomp_framework}
In this section, the basic operations in the digital AirComp framework are described as follows. Their novelty lies in the seamless integration of AirComp, digital modulation, and bit-slicing.

\begin{figure*}
    \centering
    \begin{subfigure}[b]{0.9\linewidth}
        \centering
        \includegraphics[width=\linewidth]{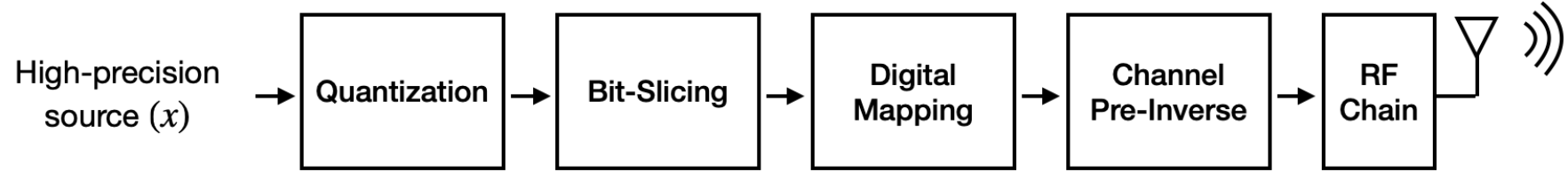}
        \caption{Transmitter operations}
        \label{fig:transmitter_design}
    \end{subfigure}
    \begin{subfigure}[b]{0.9\linewidth}
        \centering
        \includegraphics[width=\linewidth]{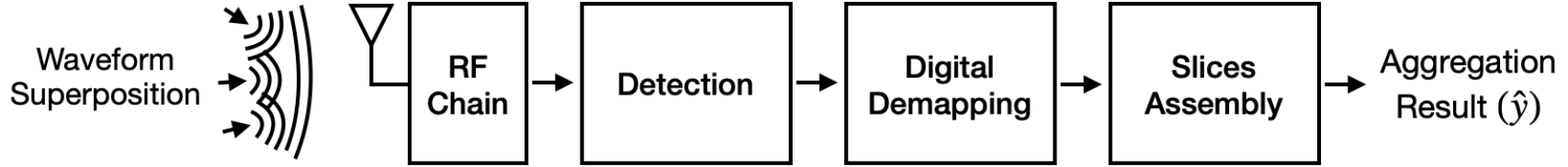}
        \caption{Receiver operations}
        \label{fig:receiver_design}
    \end{subfigure}
    \caption{Communication system employing digital AirComp with bit-slicing}
    \label{fig:transmitter_receiver_design}
\end{figure*}

\subsection{Bit Operations and Transmission}
The transmitting pre-processing operations as indicated in Fig. \ref{fig:transmitter_design} are elaborated as follows.
\subsubsection{Data Quantization}
\label{sec:quantization}
To rein in communication overhead, the values output by the analog data source should be quantized using a $B$-bit quantizer whose resolution, $B$ (e.g. 8-10 bit), is much lower than the original precision.
Consider uniform quantization of an arbitrary analog symbol, $x$, into its digital counterpart, denoted as $\hat{x}$. Let $q$ denote the index of $\hat{x}$ in the quantization codebook. Then, $q = \left\lfloor\frac{x}{\Delta}\right\rfloor$, with the step size $\Delta = (x_{\max}-x_{\min})/2^B$. Correspondingly, $\hat{x}=q\Delta+\left(\frac{1}{2}\Delta+x_{\min}\right)$.
The quantization operation introduces irreversible distortion to $x$ when converting it to $\hat{x}$.
Denoting this distortion as $\epsilon$, then, $\epsilon = x - \hat{x}$.
$\epsilon$ is accurately modeled as additive i.i.d uniform random variable independent of $x$. The distribution is given as $\epsilon\sim\mathcal{U}(-\frac{\Delta^2}{12},\frac{\Delta^2}{12})$ \cite{gray_quantization_1998}.

\subsubsection{Bit-Slicing}
\label{sec:bit-slicing}
Each quantized digital data value is sliced into multiple digital values so as to reduce their constellation sizes, increase energy per bit, thereby reducing the aggregation error rate.
Consider an arbitrary digital data value $\hat{x}$ and its associated codebook index, $q$. Let the bit sequence representing $q$ be divided into $L$ bit segments, with each representing an integer, termed \emph{sliced integer}.
Denote the bit-slicing scheme as $\mathbf{b}=[b_1,b_2,...,b_L]$, with each element denoting the number of bits of the corresponding sliced integer.
Then the $\ell$-th sliced integer, denoted as $q[\ell]$, can be related to $q$ by the following equation:
\begin{equation}
    \label{eq:bit-slicing}
    q[\ell] = \left\lfloor\frac{q}{2^{c_{\ell-1}}}\right\rfloor-2^{b_\ell}\left\lfloor\frac{q}{2^{c_\ell}}\right\rfloor,
\end{equation}
where $c_\ell=\sum_{j=1}^{\ell}b_j$ denotes the cumulated bit width of the first $\ell$ slices, with $c_0=0$.
For example, given the precision $B=3$ and the bit-slicing scheme $\mathbf{b}=[1,2]$, $q$ is sliced into two integers, $q[1]$ and $q[2]$, with $q[1]\in\{0,1\}$ and $q[2]\in\{0,1,2,3\}$.

Given the sliced integers, $\{q[\ell]\}$, the original quantized data value, $\hat{x}$, can be reconstructed as:
\begin{equation}
    \hat{x} = \sum_{\ell=1}^{L}2^{c_{\ell-1}}q[\ell]\Delta + \frac{1}{2}\Delta + x_{min}.
\end{equation}

\subsubsection{Digital Mapping}
For optimal power efficiency and compatibility with prevalent communication protocols, sliced integers are mapped to square QAM symbols for transmission.
Specifically, the $\ell$-th sliced integers of two consecutive data values, namely $q[2n-1,\ell]$ and $q[2n,\ell]$, are mapped to the $I$ and $Q$ branches of the $(n,\ell)$-th modulated symbol, denoted as $m[n,\ell]$. Mathematically, $m[n,\ell]$ can be expressed as:
\begin{equation}
    \label{eq:digitalMap}
    \begin{aligned}
    m[n,\ell]&=\left(q[2n-1,\ell]-\frac{2^{b_\ell}-1}{2}\right)d_\ell\\
    &\indent+i\left(q[2n,\ell]-\frac{2^{b_\ell}-1}{2}\right)d_\ell,
    \end{aligned}
\end{equation}
based on a $4^{b_\ell}$-QAM constellation, where, $d_\ell$ represents the distance between two adjacent constellation points. Given normalized symbol power,
\begin{equation}
    \label{eq:constellationspacing}
    d_\ell=\sqrt{\frac{6}{4^{b_\ell}-1}}.
\end{equation}
Consequently, $N\times L$ modulated symbols are generated for $M=2N$ data symbols.

Continuing the previous example, a pair of $3$-bit quantized data symbols, $\hat{x}[2n-1]$ and $\hat{x}[2n]$ (corresponding quantization indices: $q[2n-1]$ and $q[2n]$), originally require a $64$-ary QAM constellation for transmission. Bit-slicing generates a $1$-bit sliced integer pair, $(q[2n-1,1],q[2n,1])$, and a $2$-bit pair, $(q[2n-1,2],q[2n,2])$, which can be transmitted using $4$-QAM and $16$-QAM, respectively. Consequently, the originally dense constellation is replaced by two sparse constellations, used over two consecutive symbol slots, for improved reliability. The example is illustrated in Fig. \ref{fig:bitslicingillustration}.

\begin{figure}
    \centering
    \includegraphics[width=\linewidth]{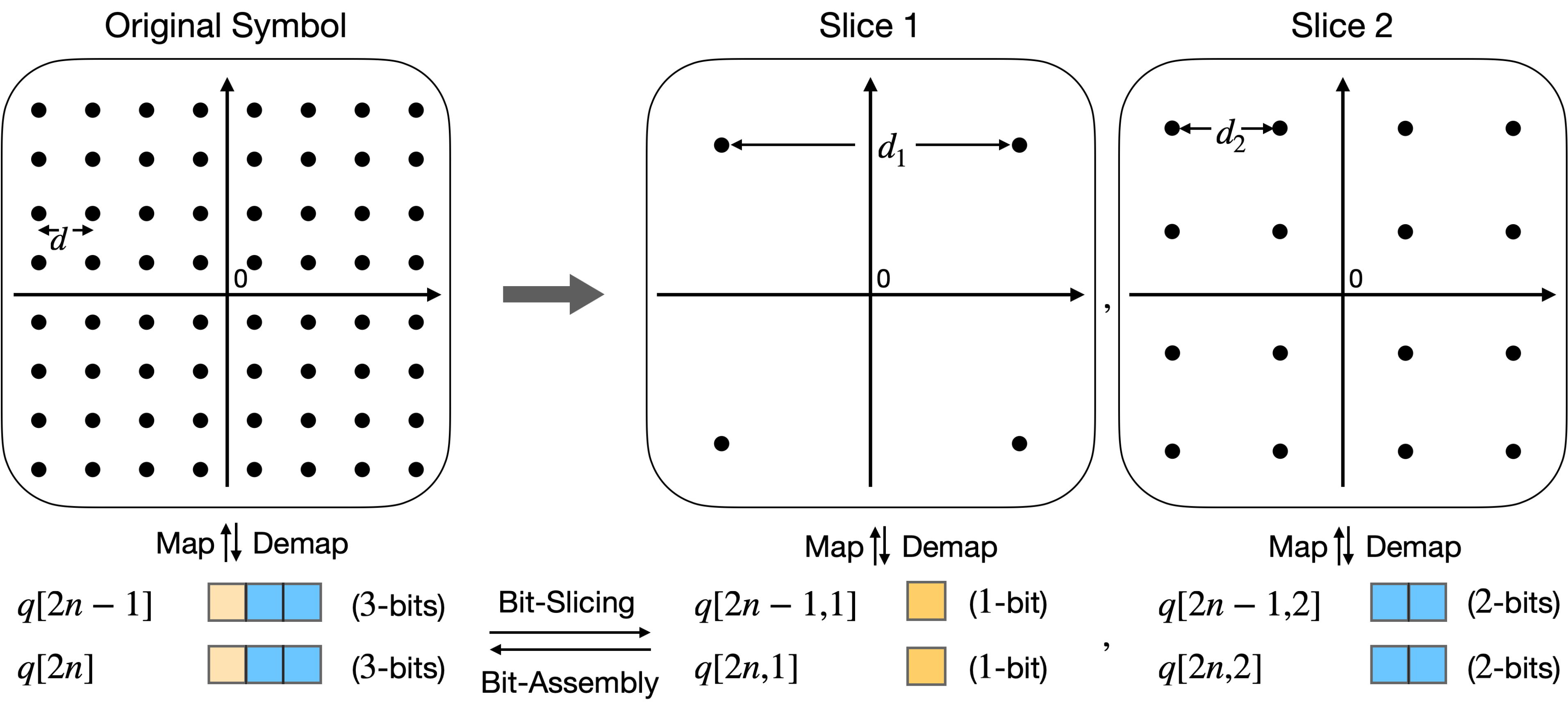}
    \caption{Illustration of bit-slicing framework and digital mapping. High bit-width integers are sliced into low bit-width integers, for transmission with sparse constellations.}
    \label{fig:bitslicingillustration}
\end{figure}

\subsubsection{Multi-Symbol Transmission}
All $K$ edge devices simultaneously transmit $N\times L$ digitally modulated symbols.
Let $s[n,\ell]=\sum_{k=1}^{K}m_k[n,\ell]$ represent the ideal superposition outcome for the $(n,\ell)$-th transmitted symbols.
Based on \eqref{eq:channel}, the $(n,\ell)$-th received symbol can be expressed as $r[n,\ell]=\rho s[n,\ell]+z[n,\ell]$.

\subsection{Receiver Post-processing}
The receiver post-processing operations are illustrated in Fig. \ref{fig:receiver_design} and described as follows.

\subsubsection{Symbol Detection and Demapping}
Consider an arbitrary received symbol, $r$. The receiver aims to estimate the corresponding ideal superimposed symbol, $s$.
Due to superposition of $K$ constellations, $s$ belongs to a $((2^{b}-1)K+1)^2$-ary QAM constellation, denoted as $\mathcal{S}$. The estimation is denoted as $\hat{s}$. The optimal detector is presented in Section~\ref{sec:optimal_detection}.

Based on \eqref{eq:digitalMap}, the $I$ and $Q$ branches of $\hat{s}$ convey two summations of sliced integers. The demapping for the real and imaginary parts of $\hat{s}$ can be performed separately.
Let the $K$-user aggregation of the $(m,\ell)$-th sliced integers be denoted as $u[m,\ell]=\sum_{k=1}^{K}q_k[m,\ell]$. Then its estimate by demapping can be expressed as:
\begin{equation}
    \label{eq:digitalDemap}
    \hat{u}[m,\ell]=
    \begin{cases}
    \Re\left(\hat{s}\left[\frac{m+1}{2},\ell\right]\right)/d_\ell+\frac{2^{b_\ell}-1}{2}K,\ \text{for odd } m,\\
    \Im\left(\hat{s}\left[\frac{m}{2},\ell\right]\right)/d_\ell+\frac{2^{b_\ell}-1}{2}K,\ \text{for even } m.\\
    \end{cases}
\end{equation}

\subsubsection{Slice-Assembly} 
To estimate the aggregation of data symbols, the receiver assembles the estimated aggregated slices, $\{\hat{u}[m,\ell]\}_{\ell=1}^L$. 
Let $u[m]=\sum_{k=1}^{K}q_{k}[m]$ denote the ideal aggregated codebook indices. Then, $u[m]$ is generated base on \eqref{eq:bit-slicing}:
\begin{equation}
    \label{eq:bitassembly}
    \hat{u}[m]=\sum_{\ell=1}^{L}\hat{u}[m,\ell]2^{c_{\ell-1}}.
\end{equation}
Finally, the desired AirComp result, $\hat{y}[m]=\sum_{k=1}^{K}\hat{x}_k[m]$, is obtained by denormalizing $\hat{u}[m]$: 
\begin{equation}
    \label{eq:computationresult}
    \hat{y}[m]=\hat{u}[m]\Delta+\left(\frac{\Delta}{2}+x_{\min}\right)K.
\end{equation}

\section{Optimal Detection for Digital AirComp}
\label{sec:optimal_detection}
For ease of notation, we omit the indices $n$, $m$ and $\ell$ in this section whenever no confusion is caused.
Optimal detection of the superimposed symbol $s$ is provided by the MAP detector, which estimates based on the prior probability $\Pr\{s=s_{\ell}\}$ and the channel noise variance. 
In this section, we first recapture the preliminaries of MAP detection. Then we derive the prior probability of the superimposed constellation points. Finally, the MAP detection algorithm for the digital AirComp receiver is presented.

\subsection{Definition of MAP Detector for Digital AirComp}
In traditional point-to-point digital communication systems, the transmitted symbols are assumed to be equiprobable. 
Under this assumption, the MAP detector can be simplified to the maximum likelihood (ML) detector, which generates decision boundaries at midpoints of adjacent constellation points.
The same approach does not retain its optimality in the more sophisticated case of AirComp, which involves aggregation over a multi-access channel.
While the transmitted symbols, $\{m_k\}$, are equiprobable, the superimposed symbols, $\{s\}$, are no longer uniformly distributed.
As the ML detector loses its optimality, we design the optimal AirComp detector based on the MAP rule defined as:
\begin{equation}
    \label{eq:digitaldetection}
    \hat{s}(r) = \underset{s_{j}\in\mathcal{S}}{\mathrm{argmax}}\{p_{j} f(r|s_{j})\},
\end{equation}
where the prior probability $p_{j}=\Pr\{s=s_{j}\}$ with $s_{j}\in\mathcal{S}$, and $f(r|s)$ is the conditional channel transition probability defined as:
\begin{equation}
    \label{eq:channeltransfer}
    f(r|s) = \frac{1}{\pi\sigma_z^2}\exp\left(-\frac{\left|r-\rho s\right|^2}{\sigma_z^2}\right).
\end{equation}

In the current system, the aggregated symbols, $\{s\}$, have independent $I$ and $Q$ branches. Then, the MAP detection of $s$ can be decoupled into two independent tasks, each concerning the detection of a pulse amplitude modulated (PAM) symbol from its real/imaginary part, denoted as $s^{(I)}$ or $s^{(Q)}$ (see Appendix~\ref{app:QAMdetection} for more details).
This allows us to focus on the study of MAP detection of PAM symbols in the rest of this section.

\subsection{Distribution of Aggregated Symbols}
\label{sec:symboldist}
Consider an arbitrary aggregated PAM symbol, $s^{(I)}$ or $s^{(Q)}$, with the superscript omitted for brevity.
By slight abuse of notation, let this symbol be denoted as $s$, and $\hat{s}$, $\mathcal{S}$, $p_{j}$ and $r$ denote the corresponding estimated symbol, constellation, prior probability and received symbol, respectively.
To derive the MAP detector for $s$, its distribution is characterized in this subsection.

The transmitted PAM symbols, $\{m_k\}$, are generated by linear mapping of $\{q_k\}$, and the bit-slicing and quantization operations preserve the data uniformity \cite{tezuka2012uniform}.
Therefore, $\{m_k\}$ are i.i.d uniform on the PAM constellation, which contains $P=2^{b}$ evenly spaced real values with a minimum distance $d=\sqrt{\frac{6}{2^b-1}}$.
Consequently, the constellation of $s$, $\mathcal{S}$, is a set of $((P-1)K+1)$ evenly spaced real values with the same spacing.
By analyzing the sum distribution using probability-generating functions, the exact distribution of $s$ is characterized in the following lemma.

\begin{Lemma}
    \label{lemma:exactdist}
    The superimposed PAM symbols follow a lattice distribution with probability given by the polynomial coefficient:
    \begin{equation}
        \label{eq:exactdist}
        \Pr\{s=s_{j}\}=
        \begin{cases}
            \left(\frac{1}{P}\right)^{K}\binom{j-1}{K}_{P}, &\text{for $s_j\in \mathcal{S}$}\\
            0, &\text{otherwise}
        \end{cases},
    \end{equation}
    where $\binom{n}{k}_m$ denotes the following polynomial coefficient \cite{steffen_eger_stirlings_2014}:
    \begin{equation}
        \label{eq:polycoef}
        \binom{n}{k}_{m}=\sum_{t=0}(-1)^t\binom{k}{t}\binom{n+k-tm-1}{k-1},
    \end{equation}
    with $\binom{n}{k}$ being the binomial coefficient.
\end{Lemma}
\begin{proof}
    (See Appendix \ref{app:exactdist}).
\end{proof}
\begin{Remark}
    An intuitive way to understand Lemma~\ref{lemma:exactdist} is by noticing its similarity to a classic combinatorial problem---balls-in-bins. We can write the prior probability, $p_{j}$, as given in \eqref{eq:exactdist}, as $p_{j}=\frac{\text{\# of combinations of $\{m_{k}\}$ generating $s_{j}$}}{\text{\# of all possible combinations of $\{m_{k}\}$}}$. The numerator is the same as the number of possible ways to distribute $(j-1)$ balls into $K$ bins, where each bin can hold at most $(P-1)$ balls. Such number can be calculated as $\binom{j-1}{K}_{P}$. The denominator counts all possible ways to distribute the balls, which is $P^K$.
\end{Remark}

However, the distribution in \eqref{eq:exactdist} is too complex to generate insightful results on the MAP detection boundaries. 
For tractability, we propose to pursue the Normal approximation as follows.
Since $\{m_{k}\}$ are random variables having lattice distribution with finite mean $0$ and variance $\frac{1}{\sqrt{2}}$. The local limiting theorem \cite{b_v_gnedenko_limit_1954} states that the distribution of $s$ can be approximated by the Normal pdf asymptotically, given by:
\begin{equation}
    \label{eq:local_limit}
    \lim_{K\to\infty}\sup_{s_j\in\mathcal{S}}\left|\frac{\sqrt{K}}{\sqrt{2}d}p_{j}-\frac{1}{\sqrt{2\pi }}\exp\left(-\frac{(jd-\frac{P+1}{2}Kd)^2}{K}\right)\right|=0.
\end{equation}
It follows that the distribution of aggregated PAM symbols as specified in Lemma~\ref{lemma:exactdist} can be approximated by the discrete Normal distribution on lattice given as
    \begin{equation}
        \label{eq:distapprox}
        \Pr\{s=s_{j}\}\approx
        \begin{cases}
            \frac{1}{\sqrt{2\pi \sigma_{j}^2}}\exp\left(-\frac{(j-\mu_{j})^2}{2\sigma_{j}^2}\right), &\text{for $s_j\in\mathcal{S}$} \\
            0, &\text{otherwise},
        \end{cases},
    \end{equation}
where $\mu_{j}=\frac{P+1}{2}K$ and $\sigma_{j}^2=\frac{P^2-1}{12}K$.
The Normal approximation is close to the exact distribution even for small $K$ (e.g., $K=4$) as shown in Fig. \ref{fig:symboldistribution}. 
The approximation accuracy grows as $K$ increases.

\begin{figure}
    \centering
    \includegraphics[width=\linewidth]{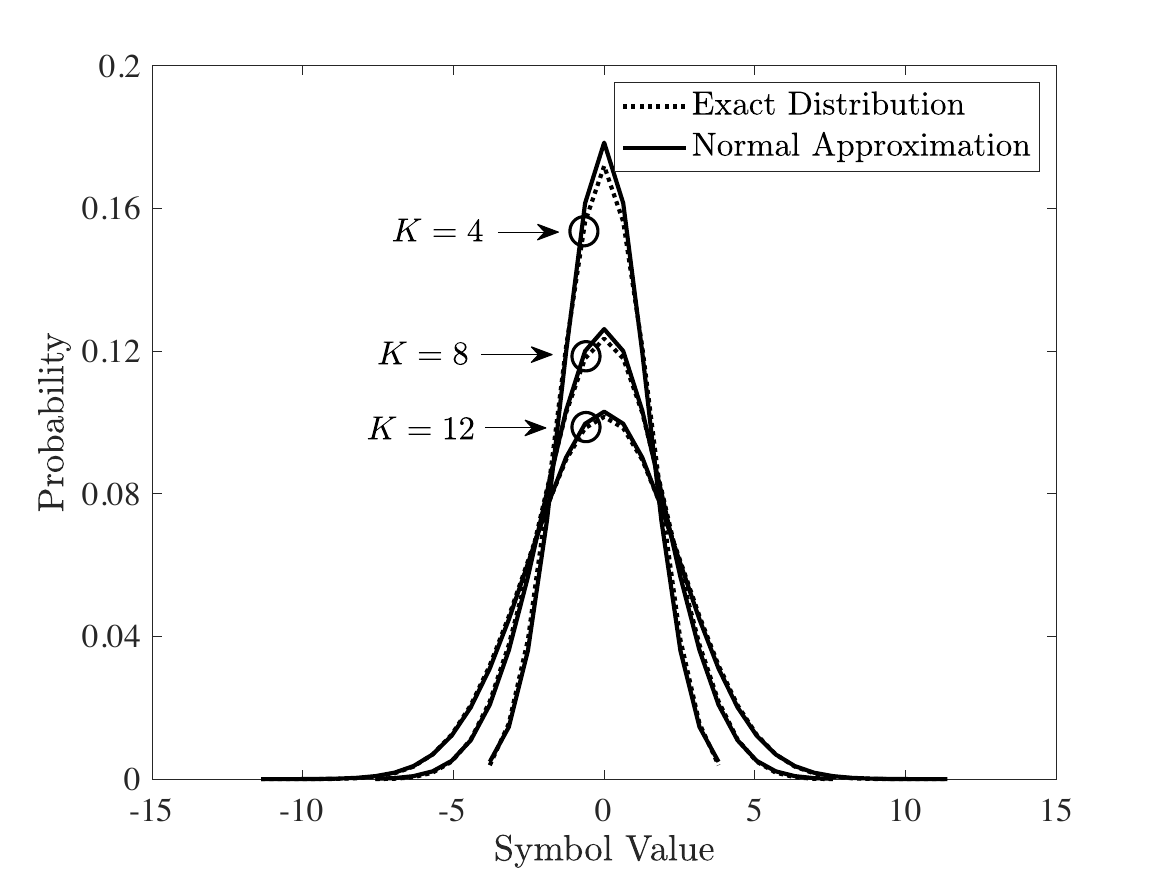}
    \caption{Comparison of the exact and approximate distribution of superimposed symbol with $P=4$ and different number of edge devices.}
    \label{fig:symboldistribution}
\end{figure}

\subsection{Optimal MAP Detector}
The optimal MAP detection in \eqref{eq:digitaldetection} can be specified by defining detection regions, $\{\mathcal{R}_{j}\}$, corresponding to the constellation points, $\{s_{j}\}$.
Each region, say $\mathcal{R}_{j}$, is defined as
\begin{equation}
    \label{eq:detectionregion}
    \mathcal{R}_{j}=\left\{r\in\mathbb{R}:p_{j} f(r|s_{j})\geq p_{m}f(r|s_{m}),\ \forall s_{m}\in\mathcal{S}\right\}.
\end{equation}

For $\{s_{j}\}_{j=2}^{(P-1)K}$, i.e. the interior symbols of the PAM constellation, a decision region, $\mathcal{R}_{j}$, can be expressed by a pair of hard decision boundaries located between $s_{j}$ and its left and right neighbors, denoted as $b_{j}^-$ and $b_{j}^+$, respectively.
The definition can be extended to the extreme points, $s_{1}$ and $s_{(P-1)K+1}$, by setting $b_{1}^-=-\infty$, and $b_{(P-1)K+1}^+=\infty$. It follows that \eqref{eq:detectionregion} can be rewritten as:
\begin{equation}
    \label{eq:PAMdetectionregion}
    \mathcal{R}_{j}=\left\{r\in\mathbb{R}:b_{j}^-< r<b_{j}^+\right\},
\end{equation}
with $b_{j}^-=b_{j-1}^+$.

\begin{Proposition}
    \label{prop:PAMdecisionboundary}
    (Optimal Decision Boundaries) The MAP decision boundaries for the aggregated PAM constellation are given as
    \begin{equation}
        \label{eq:decisionboundary}
        b_{j}^-=\rho\left(s_{j}-\frac{d}{2}\right)\left(1+\frac{1}{\gamma K}\right),\ b_{j}^-=b_{j-1}^+,
    \end{equation}
    where $\gamma=\rho^2/\sigma_z^2$ denotes the SNR for an individual signal.
\end{Proposition}
\begin{proof}
    By substituting \eqref{eq:channeltransfer} to \eqref{eq:detectionregion}, we solve the decision boundary $b_{j}^-$ by:
    \begin{equation}
        \label{eq:decisionboundarynodist}
        b_{j}^-=\rho\left(s_{j}-\frac{d}{2}\right)-\frac{\sigma_z^2}{2\rho d}\log\frac{p_{j}}{p_{j-1}}.
    \end{equation}
    Substituting \eqref{eq:distapprox} into \eqref{eq:decisionboundarynodist} and note that $s_j=(j+K-1-\mu_j)d$ gives the desired result.
\end{proof}

\begin{figure}
    \centering
    \includegraphics[width=\linewidth]{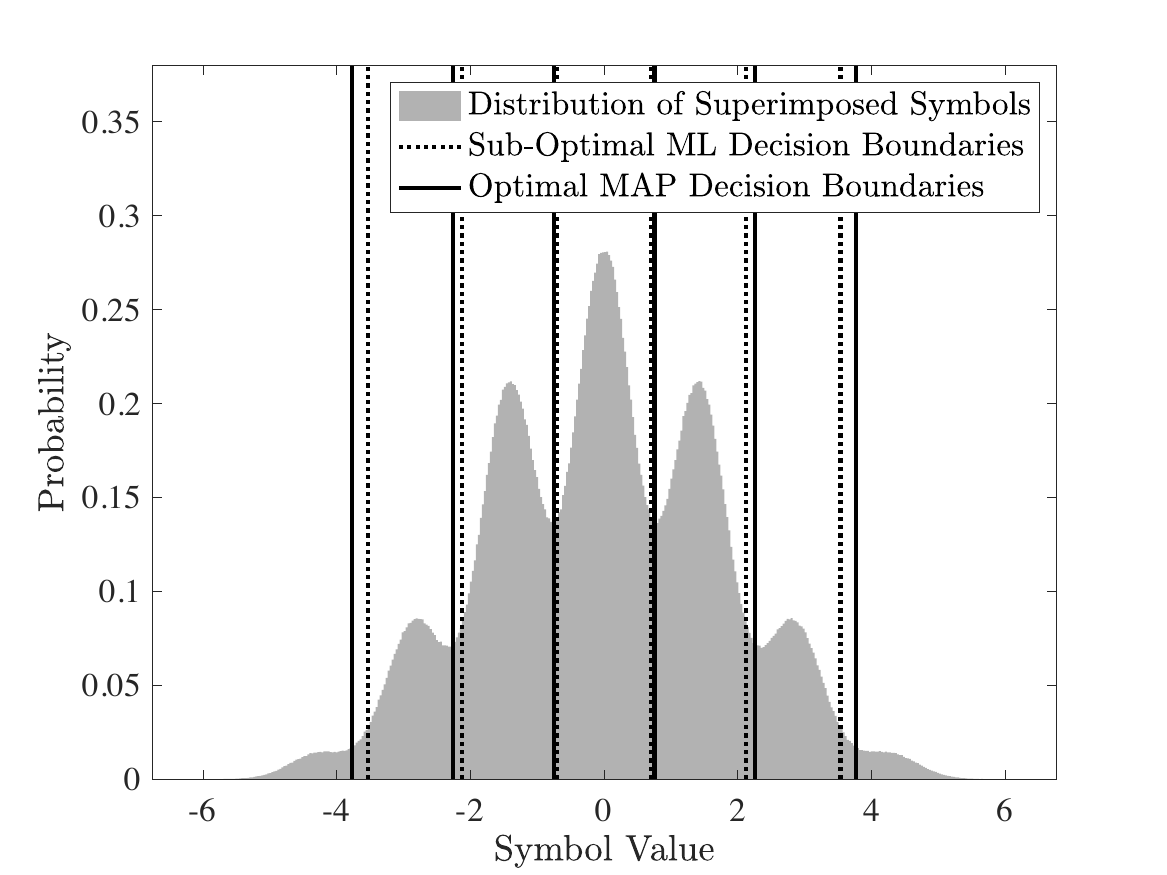}
    \caption{Received symbol distribution and corresponding ML and MAP decision boundaries at $4\text{ dB}$ SNR for $P=2$ and $K=6$.}
    \label{fig:detection}
\end{figure}

Proposition~\ref{prop:PAMdecisionboundary} suggests that $b_{j}^-$ is located with an SNR-related offset, $\frac{1}{\gamma K}$, from the midpoint of $\rho s_{j-1}$ and $\rho s_{j}$. As illustrated in Fig.~\ref{fig:detection}, the offset controls the decision boundary to be closer to the constellation point with smaller prior probability. In other words, the MAP detector tends to detect the aggregated symbol as the more probable constellation point. 
It is also interesting to notice that the ratio between $b_{j}^-$ and $\rho\left(s_{j}-d/2\right)$ (i.e., the midpoint between $\rho s_{j-1}$ and $\rho s_{j}$, which is the ML decision boundary) remains constant for different $j$. However, as the value of $j$ increases, the absolute difference between the ML and MAP decision boundaries increases, same as observed from Fig.~\ref{fig:detection}.
Additionally, since the offset is inversely proportional to the SNR, the MAP decision boundaries approach their ML counterparts as the SNR increases. 
Finally, for a sanity check, the MAP decision boundaries in \eqref{eq:decisionboundary} reduce to their ML counterparts if we set prior probability as uniform.

\section{Digital or Analog AirComp}
\label{sec:error_analysis}
In this section, a comparative analysis is conducted on the error performance of digital with respect to the traditional analog AirComp. Repetition coding is applied to the latter to ensure fair comparisons under the constraint of identical channel uses. Then, conditions favoring digital AirComp are derived.

\subsection{Error Analysis for Analog AirComp with Repetition Coding}
\subsubsection{Analog AirComp with Repetition Coding}
\label{sec:analogscheme}
We consider the simple method of repetition coding to improve the reliability of analog AirComp via repeatedly transmitting each complex data symbol over multiple symbol slots, i.e., $L=R/N$ slots, assuming $L$ is an integer. Then averaging of the received aggregated symbols suppresses the channel noise to improve AirComp reliability.
For complex analog mapping, each pair of analog symbols, say $(x[2n-1],x[2n])$, is mapped to $L$ identical complex symbols $\{m[n,\ell]\}$, generating $N\times L$ symbols in total. Mathematically,
\begin{equation}
\label{eq:analogmap}
    m[n,\ell]=\left(x[2n-1]-\frac{\tau_x}{2}\right)\frac{\sqrt{6}}{\delta_x}+i\left(x[2n]-\frac{\tau_x}{2}\right)\frac{\sqrt{6}}{\delta_x},
\end{equation}
where $\tau_x=x_{\max}+x_{\min}$, and $\delta_x = x_{\max}-x_{\min}$. 
All devices simultaneously transmit the mapped analog symbols over the multi-access channel.
As a result, the received aggregated symbol corresponding to the $(n,\ell)$-th transmitted symbols is given as:
\begin{equation}
    \label{eq:analog_aggregated_symbol}
    r[n,\ell]=\rho\sum_{k=1}^{K}m_{k}[n,\ell]+z[n,\ell].
\end{equation}

The receiver first demaps $r[n,\ell]$ to estimate the  aggregation of transmitted symbols, denoted as $\hat{s}[m,\ell]$, as follows
\begin{equation}
    \label{eq:analogdemap}
    \hat{s}[m,\ell]=
    \begin{cases}
    \Re\left(r\left[\frac{m+1}{2},\ell\right]\right)\frac{\delta_x}{\rho\sqrt{6}}+\frac{\tau_x}{2}K,\ \text{for odd } m,\\
    \Im\left(r\left[\frac{m}{2},\ell\right]\right)\frac{\delta_x}{\rho\sqrt{6}}+\frac{\tau_x}{2}K,\ \text{for even } m.\\
    \end{cases}
\end{equation}
Then $\{\hat{s}[m,\ell]\}$ are averaged over $L$ repetitions to suppress channel noise and generate the desired AirComp result of the $m$-th data values, by: $\hat{y}[m]=\frac{1}{L}\sum_{\ell=1}^{L}[\hat{s}[m,\ell]]$.

\subsubsection{Analog AirComp Error}
The error of analog AirComp with repetition coding can be obtained using the definition in \eqref{eq:MSE} as:
\begin{equation}
    \label{eq:analogerror}
    \mathcal{E}^{(ana)}=\frac{\delta_x^2}{12L}\frac{1}{\gamma}.
\end{equation}
One can observe that repetition coding suppresses error by a factor of $L$. 
Moreover, $\mathcal{E}^{(ana)}$ is inversely proportional to the transmit SNR, $\gamma$. 

\subsection{Error Analysis for Digital AirComp}
\label{sec:digitalaircomperror}
Based on the definition in \eqref{eq:MSE}, the digital AirComp error, denoted as $\mathcal{E}^{(dig)}$, can be obtained as:
\begin{align}
    \mathcal{E}^{(dig)}&=\mathbb{E}\left[\left(\Delta\hat{u}+\left(\frac{\Delta}{2}+x_{\min}\right)K\right.\right.\notag\\
    &\left.\left.\indent-\sum_{k=1}^{K}\left(q\Delta+\frac{\Delta}{2}+x_{\min}+\epsilon\right)\right)^2\right], \\
    &=\Delta^2\mathbb{E}\left[\left(\hat{u}-u\right)^2\right]+K\mathbb{E}\left[\epsilon^2\right]\triangleq\mathcal{E}_{agg}+\mathcal{E}_{qua}, \label{eq:digitalerror}
\end{align}
where \eqref{eq:digitalerror} follows from the independence between $\hat{u}$ and $\epsilon$, and the i.i.d. property of $\epsilon$ as described in Section~\ref{sec:quantization}.
It is observed that $\mathcal{E}^{(dig)}$ consists of two parts. 
The first part is AirComp error assuming no quantization, termed the \emph{aggregation error}, and denoted as $\mathcal{E}_{agg}\triangleq\Delta^2\mathbb{E}\left[\left(\hat{u}-u\right)^2\right]$. 
The second part is the \emph{quantization error}, denoted as $\mathcal{E}_{qua}\triangleq K\mathbb{E}\left[\epsilon^2\right]$. With the assumption on uniform quantization, $\mathcal{E}_{qua}=\frac{\Delta^2}{12}$.

The aggregation error, $\mathcal{E}_{agg}$, is the focus of the following analysis. 
Since an arbitrary $\hat{u}$ is assembled from $L$ aggregated sliced integers, $\{\hat{u}[\ell]\}$, by \eqref{eq:bitassembly}, which are demapped from detected symbols $\{\hat{s}[\ell]\}$ by \eqref{eq:digitalDemap}. 
It follows that:
\begin{align}
    \mathcal{E}_{agg}&=\Delta^2\mathbb{E}\left[\left(\sum_{\ell=1}^{L}\hat{u}[\ell]2^{c_{\ell-1}}-\sum_{\ell=1}^{L}u[\ell]2^{c_{\ell-1}}\right)^2\right],\\
    &=\Delta^2\sum_{\ell=1}^{L}4^{c_{\ell-1}}\mathbb{E}\left[\left(\hat{u}[\ell]-u[\ell]\right)^2\right]\label{eq:aggregationerror1},\\
    &=\Delta^2\sum_{\ell=1}^{L}\frac{4^{c_{\ell-1}}}{d_\ell^2}\mathbb{E}\left[\left(\hat{s}[\ell]-s[\ell]\right)^2\right],\label{eq:aggregationerror2}
\end{align}
where \eqref{eq:aggregationerror1} follows from the independence of $\{\left(\hat{u}[\ell]-u[\ell]\right)\}$.
The result in \eqref{eq:aggregationerror2} shows that $\mathcal{E}_{agg}$ is the weighted sum of the expected symbol detection error, defined as $\mathcal{E}_{det}\triangleq\mathbb{E}\left[\left(\hat{s}[\ell]-s[\ell]\right)^2\right]$.

\begin{Lemma}
    \label{lemma:detectionerror}
    The expected error for the detection of aggregated symbols is given as:
    \begin{equation}
        \label{eq:detectionerror}
        \mathcal{E}_{det}=\sum_{s_j\in\mathcal{S}}p_{j}\sum_{s_m\in\mathcal{S}}(s_m-s_j)^2P_{j\to m}(\sigma_z),
    \end{equation}
    where $p_{j}$ follows from \eqref{eq:distapprox} and $P_{j\to m}(\sigma_z)$ denotes the pairwise error probability that the symbol $s_{j}$ is detected as $s_{m}$, is given as:
    \begin{align}
        P_{j\to m}(\sigma_z)&=\Pr\{\hat{s}=s_{m}|s=s_{j}\}
        \label{eq:pairwiseerrordefinition},\\
        &=Q\left(\frac{b_{m}^--\rho s_{j}}{\frac{\sigma_z}{\sqrt{2}}}\right)-Q\left(\frac{b_{m}^+-\rho s_{j}}{\frac{\sigma_z}{\sqrt{2}}}\right),\label{eq:pairwiseerror}
    \end{align}
    where $Q(x)$ is the Q-function defined as $Q(x)=\frac{1}{\sqrt{2\pi}}\int_x^\infty e^{-\frac{t^2}{2}}dt$.
\end{Lemma}

It should be emphasized that \eqref{eq:detectionerror} reduces to the symbol error rate for traditional digital modulation if we set $(s_m-s_j)=1, \forall m\neq j$, so that any detection error is treated as single symbol error with equal weight.
In other words, the errors in terms of data values, $\{\left(s_m-s_j\right)^2\}$, and weights, $\{p_{j}\}$, differentiate digital AirComp from traditional digital modulation.
Following \eqref{eq:aggregationerror2} and via derivation of the asymptotic expression of \eqref{eq:detectionerror}, we characterize $\mathcal{E}_{agg}$ as follows.
\begin{Lemma}
    \label{lemma:aggregationerror}
    The aggregation error satisfies:
    \begin{equation}
        \label{eq:aggregationerror}
        \lim_{\sigma_z\to0}\mathcal{E}_{agg} = \Delta^2\sum_{\ell=1}^{L}4^{c_{\ell-1}}A_\ell Q\left(\frac{\rho d_\ell}{\sqrt{2}\sigma_z}\right),
    \end{equation}
    where the term $A_\ell\triangleq\sum_{s_j\in\mathcal{S}_\ell}\left(\sqrt{p_{j-1}p_{j}}+ \sqrt{p_{j} p_{j+1}} \right)$ is defined to simplify the notation.
\end{Lemma}
\begin{proof}
    (See Appendix \ref{app:detectionerror}).
\end{proof}
The MAP detector is closely related to the term $A_\ell$. If $\{s_j\}$ are equiprobable, for which the ML detector is optimal, this coefficient reduces to 2.
Consider the simple scheme of uniform bit allocation to slices, termed the \emph{uniform bit-slicing} scheme. Its generalization is the topic of the next section. Under this scheme, $L$ is assumed to be integer multiply of $B$, such that each sliced integer has a bit-width of $b=B/L$.

\begin{Proposition}
    \label{prop:digitalaircomperror}
    The digital AirComp error with uniform bit-slicing is given as:
    \begin{equation}
        \label{eq:digitalerroruniform}
        \mathcal{E}^{(dig,u)} = \Delta^2ACQ\left(\frac{\rho d}{\sqrt{2}\sigma_z}\right) + \frac{\Delta^2}{12}K,
    \end{equation}
    where the term $C\triangleq(4^B-1)/(4^b-1)$ is defined to simplify the notation.
\end{Proposition}

\subsection{Comparison between Digital and Analog AirComp}
We combine \eqref{eq:analogerror} and \eqref{eq:digitalerroruniform} to derive as shown below the SNR regime where digital AirComp outperforms its analog counterpart to guide the selection between the two schemes. 
\begin{Theorem}
    \label{theorem:digitalvsanalog}
    The error of digital AirComp is smaller than that of the analog AirComp (i.e., $\mathcal{E}^{(dig,u)}\leq\mathcal{E}^{(ana)}$) if the SNR satisfies:
    \begin{equation}
        \label{eq:snrregime}
        \gamma \in \left[-\frac{4}{d^2}W_{r,-1}\left(-\frac{d^2 4^B}{24LAC}\right), -\frac{4}{d^2}W_{r,-2}\left(-\frac{d^24^B}{24LAC}\right)\right],
    \end{equation}
    where $W_r(\cdot)$ denotes the generalized $r$-Lambert function with $r=\frac{K}{6AC}\in(0,1/e^2)$\textup{\cite{mezo_generalization_2017}}.
\end{Theorem}
\begin{proof}
    (See Appendix \ref{app:digitalvsanalog})
\end{proof}

\begin{figure}
    \centering
    \includegraphics[width=\linewidth]{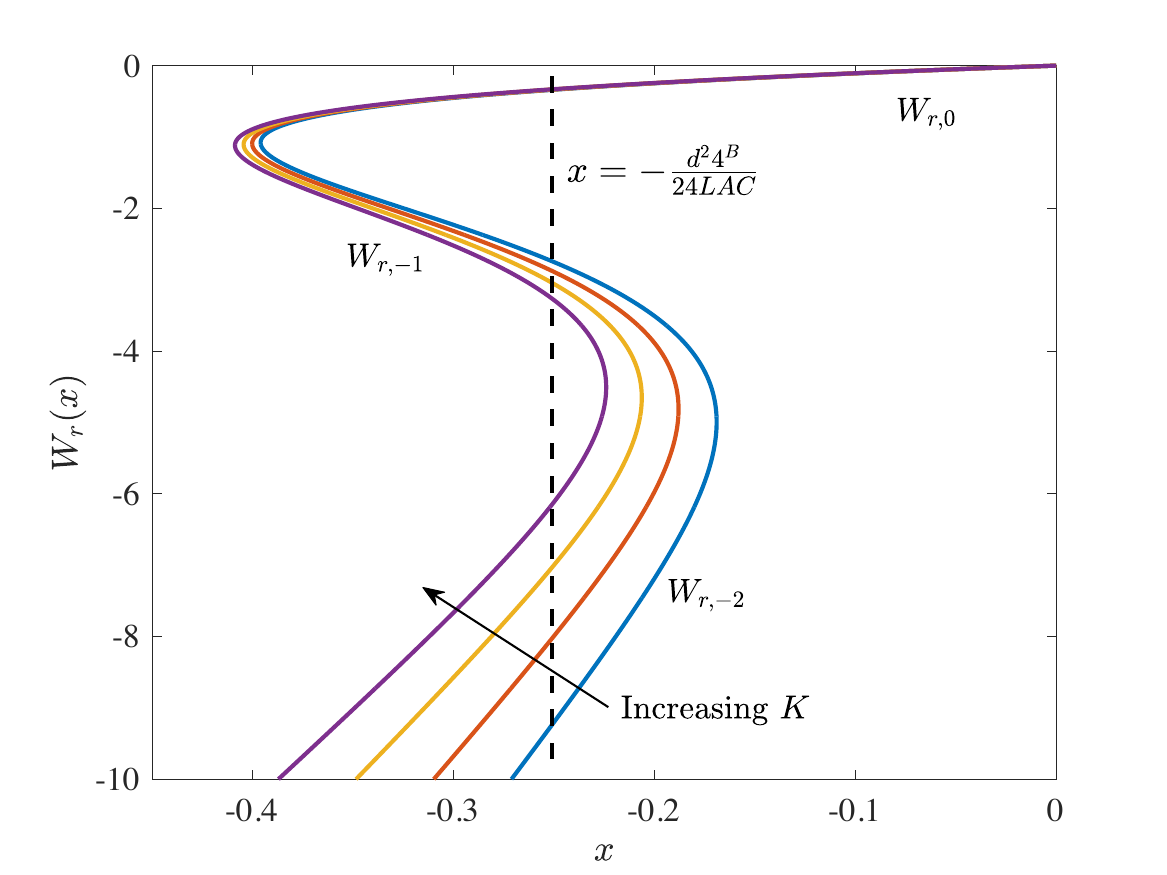}
    \caption{The effect of increasing $K$ on the $r$-Lambert function and thus the SNR regime in Theorem \ref{theorem:digitalvsanalog}.}
    \label{fig:lambertwithdiffr}
\end{figure}

Theorem \ref{theorem:digitalvsanalog} reveals that the relative performance is affected by the number of participating devices, $K$, while other parameters are fixed.
To be specific, according to the property of $W_r$, increasing K causes the left boundary of \eqref{eq:snrregime} to increase and its right boundary to reduce. 
This property is illustrated in Fig.~\ref{fig:lambertwithdiffr}, where the length of the segment of dash line between its intersection with $W_{-1}$ and $W_{-2}$ is proportional to the SNR regime given by \eqref{eq:snrregime}.
It is possible for \eqref{eq:snrregime} to reduce to an empty set as $K$ increases, which is due to the accumulated quantization error. To prevent this, $K$ must satisfy the condition given by the following Corollary.
\begin{Corollary}
    \label{corollary:limitofK}
    The SNR regime given by \eqref{eq:snrregime} is not empty if $r=\frac{K}{6AC}$ satisfies:
    \begin{equation}
        -r\left[W_{-1}\left(-re\right)+\frac{1}{W_{-1}\left(-re\right)}-2\right] < \frac{1}{8L},
    \end{equation}
    where $W_{-1}$ is the secondary branch of the standard Lambert function.
\end{Corollary}
\begin{proof}
    (See Appendix \ref{app:limitofK})
\end{proof}
The largest possible $K$ given by Corollary~\ref{corollary:limitofK} is typically a sufficiently large number for the said SNR range to be practical. For example, under the system setup of $B=10$ and $L=5$, $K\leq2344$ guarantees a non-empty SNR regime where digital AirComp is preferred.

\section{Extension: Adaptive Digital AirComp}
\label{sec:extension_adaptive}
In the preceding section, uniform bit-slicing is assumed and the quantization precision, $B$, is fixed.
In this section, building on the scheme in the preceding section, we consider the optimization of these parameters according to the heterogeneity of bit importance level and channel state, giving the name \emph{adaptive digital AirComp}.

The problem of minimizing the digital AirComp error by optimizing the bit-slicing scheme and quantization precision is formulated as follows:
\begin{equation}
    \label{eq:optimizationerror}
    \begin{aligned}
        \min_{\{B, b_\ell\}} \quad &\mathcal{E}_{agg}+\mathcal{E}_{qua}\\
        \text{s.t.} \quad &\{b_\ell\}\in\mathbb{Z}_+,\\
        &\sum_{\ell=1}^{L}b_\ell=B.\\
    \end{aligned}\tag{P1}
\end{equation}
In practice, the quantization precision, $B$, is associated with the quantizer hardware, which has limited flexibility to adjust to the channel conditions. Therefore, $B$ is optimized offline using the average SNR, denoted as $\bar{\gamma}$. 
On the other hand, the bit-slicing scheme, $\{b_\ell\}$, can be adjusted in real time based on the instantaneous SNR.
Hence, the solution to Problem~\eqref{eq:optimizationerror} is generated by solving two sub-problems: the offline optimization of $B$ and the online optimization of $\{b_\ell\}$.

The first sub-problem can be formulated as:
\begin{equation}
    \label{eq:optimizationerrorB}
    \begin{aligned}
        \min_{B} \quad &\bar{\mathcal{E}}_{agg}+\mathcal{E}_{qua}\\
        \text{s.t.} \quad &\{b_\ell\}\in\mathbb{Z}_+,\\
        &\sum_{\ell=1}^{L}b_\ell=B,\\
    \end{aligned}\tag{P2}
\end{equation}
where $\bar{\mathcal{E}}_{agg}$ is $\mathcal{E}_{agg}$ evaluated at $\gamma=\bar{\gamma}$.
One can observe that $B$ controls a trade-off between the aggregation error and the quantization error.
On one hand, the quantization error, $\mathcal{E}_{qua}$, monotonically decreases with the increase of $B$.
On the other hand, the aggregation error, $\mathcal{E}_{agg}$, tends to increase with $B$, as the increase of $B$ results in the increase of the alphabet size of the transmitted symbol, increasing the detection error rate. 
For example, suppose we use a uniform-by-best-effort (UBE) scheme for $B/L\notin\mathbb{Z}$, where we make the best effort on equalizing the bit allocation to slices, the trade-off between $\mathcal{E}_{agg}$ and $\mathcal{E}_{qua}$ for different $B$ is shown in Fig.~\ref{fig:aggregation_quantization_tradeoff}.
The standard approach to Problem \eqref{eq:optimizationerrorB} is through numerical methods such as exhaustive search. However, the complexity is prohibitive, with $\binom{B-1}{L-1}$ evaluations of the objective function for each $B$.
For this reason, we adopt a suboptimal solution by searching $B$ under the assumption of UBE bit-slicing scheme, therefore significantly reducing the search complexity to one evaluation per $B$.

\begin{figure}
    \centering
    \includegraphics[width=\linewidth]{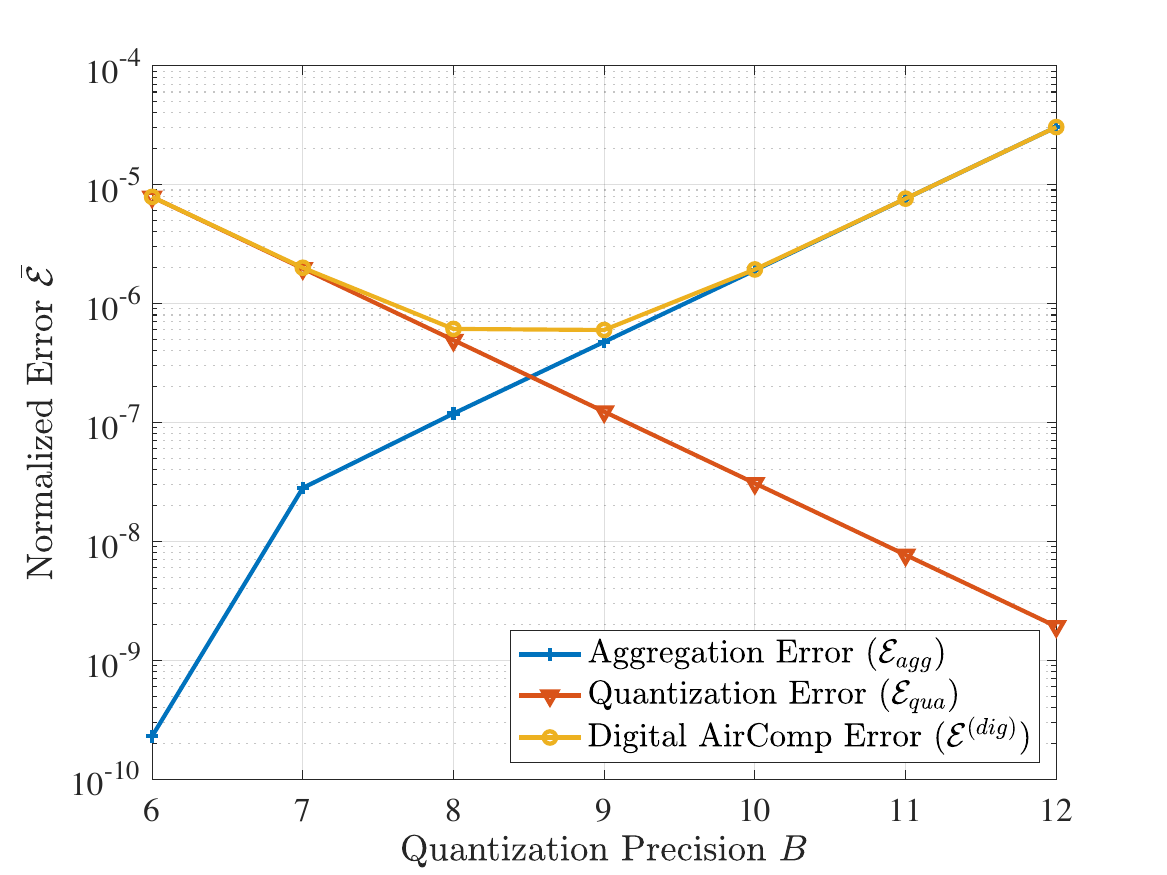}
    \caption{The trade-off between aggregation and quantization errors by varying the quantization precision for $\gamma = 15 \text{ dB}$.}
    \label{fig:aggregation_quantization_tradeoff}
\end{figure}

The second sub-problem is to optimize the bit-slicing scheme, $\{b_\ell\}$, for a given $B^\star$. This factors in the heterogeneous importance levels of bits depending on their positions in a binary representation of a data value.
It can be observed that $\mathcal{E}_{qua}$ is a constant due to the fixed quantization precision, hence the sub-problem is equivalent to minimizing $\mathcal{E}_{agg}$ as follows:
\begin{equation}
    \label{eq:optimizationerrorb}
    \begin{aligned}
        \min_{\{b_\ell\}} \quad &\mathcal{E}_{agg}\\
        \text{s.t.} \quad &\{b_\ell\}\in\mathbb{Z}_+, \\
        &\sum_{\ell=1}^{L}b_\ell=B^\star.
    \end{aligned}\tag{P3}
\end{equation}
As Problem~\eqref{eq:optimizationerrorb} is integer constrained, and $\{b_\ell\}$ are coupled, the standard approach to it is through exhaustive search. We propose to exploit the bit importance to reduce the search domain of $\{b_\ell\}$ as described in the following. 
Because bits in the binary representation of $q$ are not equally important, that is, from the more significant bit (MSE, bit with highest-order place in the binary representation) to the less significant bit (LSB, bit with the lowest-order place), the effect of bit error on the aggregation error decreases. 
Also note that for the bit-slicing operation in \eqref{eq:bit-slicing}, the leftmost slice, $b_1$, corresponds to the bits with lower significance.
Therefore, we keep $\{b_\ell\}$ in non-increasing order to prioritize the accurate aggregation of more significant bits. 
In this way, the number of objective function evaluations is reduced from $L$-composition of $B$ to $L$-partition of $B$. For example, the number of evaluations is reduced from $35$ to $5$ for $B=8$ and $L=4$.

The proposed algorithm for \emph{adaptive digital AirComp} is summarized in Algorithm~\ref{alg:adaptive_digital_aircomp}.

\begin{algorithm}
    \caption{Adaptive Digital AirComp Scheme}
    \label{alg:adaptive_digital_aircomp}
    \KwIn{$\bar{\gamma}$, $\gamma$, $L$, $K$}
    \textbf{Offline Optimization of $B$:} Search for solution to Problem~\eqref{eq:optimizationerrorB} with the UBE bit-slicing scheme to obtain $B^\star$\\
    \textbf{Online Optimization of $\{b_\ell\}$:} Search for solution to Problem~\eqref{eq:optimizationerrorb} with the constraint $b_1\geq b_2\geq\cdots\geq b_L$ to obtain $\{b_\ell^\star\}$\\
    \KwOut{$B^\star$, $\{b_\ell^\star\}$}
\end{algorithm}

\section{Numerical Results}
\label{sec:numerical_results}

\subsection{Simulation Settings}
\subsubsection{System and Communication Settings}
Consider a system with $K=10$ devices, the size of local data to be aggregated is $M=181,503$ scalar samples, which is the average pixel count for images in the ImageNet dataset \cite{deng_imagenet_2009}.
If not specified, the quantization precision to be used in the digital AirComp scheme is $B=6$, and the repetition factor $L=6$.
The channel between the devices and the server is modeled with a $\chi$-square distribution with $\kappa=1$.
At the transmission stage, the devices are assumed to follow the channel pre-inversion technique with average transmit power $p_k=1$, corresponding to $\rho=1$. The change in noise power $\sigma_z^2$ is used to control $\gamma$.
The total symbol rate is assumed to be $10^6$ symbols per second.
\subsubsection{Benchmarking Schemes}
Two benchmark schemes are used to compare with digital AirComp. 
The first scheme, termed \emph{orthogonal access}, transmits the data on devices with traditional end-to-end digital communication. The classic orthogonal access scheme, e.g., time-division multiple access (TDMA), is employed to transmit the data on different devices.
Analog data on the device is quantized with the same bit-width $B$ as in digital AirComp.
The same channel pre-inversion technique as in digital AirComp is adopted, and the quantized data bits are reliably transmitted to the server at a capacity-achieving rate of $10^6\log_2(1+\gamma)$ bits per second.
The second benchmark scheme is the \emph{analog AirComp} with repetition coding, as described in Section~\ref{sec:analogscheme}. The same repetition factor $L=6$ is used for the analog AirComp scheme as in the digital AirComp scheme.

\subsection{Performance of Digital AirComp}

\begin{figure}
    \centering
    \includegraphics[width=\linewidth]{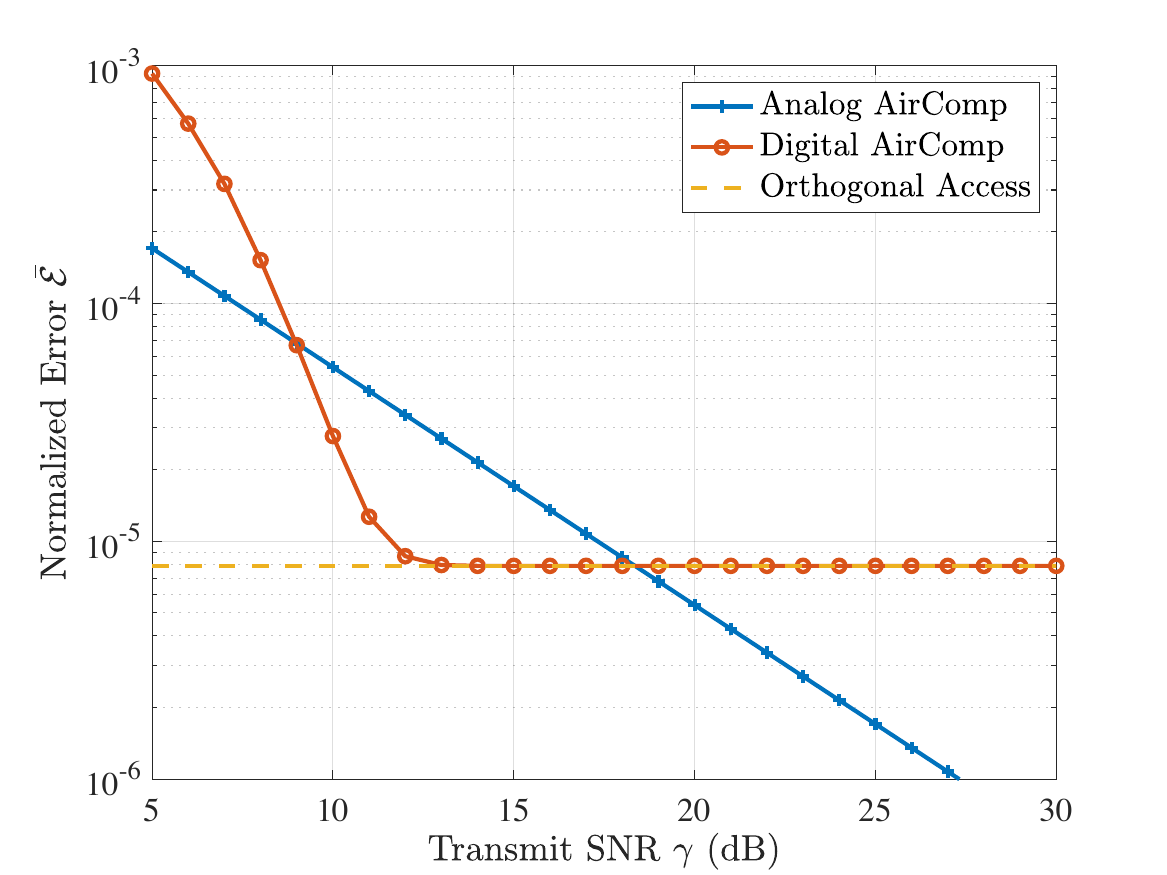}
    \caption{Comparison of data aggregation error among digital AirComp, analog AirComp and orthogonal access.}
    \label{fig:nmse_vs_snr_benchmark}
\end{figure}

We first compare the computation (i.e., data aggregation) accuracy of digital AirComp with those of the two benchmark schemes. The aggregation error is evaluated by the normalized MSE metric as defined in \eqref{eq:nMSE}. 
Fig.~\ref{fig:nmse_vs_snr_benchmark} shows that digital AirComp outperforms analog AirComp in the moderate to high SNR regime (from, approximately, $\text{9 dB}$ to $\text{18 dB}$). 
The analysis in Section~\ref{sec:error_analysis} explains this effect by revealing that the digital AirComp error decreases at an exponential rate before saturation, compared to the slow linear decrease for analog AirComp. 
One can also observe from Fig.~\ref{fig:nmse_vs_snr_benchmark} that the digital AirComp error in the high SNR regime saturates similarly as orthogonal access due to the quantization error in both digital schemes.
On the other hand, the error floor does not exist for the analog AirComp approach.

\begin{figure}
    \centering
    \includegraphics[width=\linewidth]{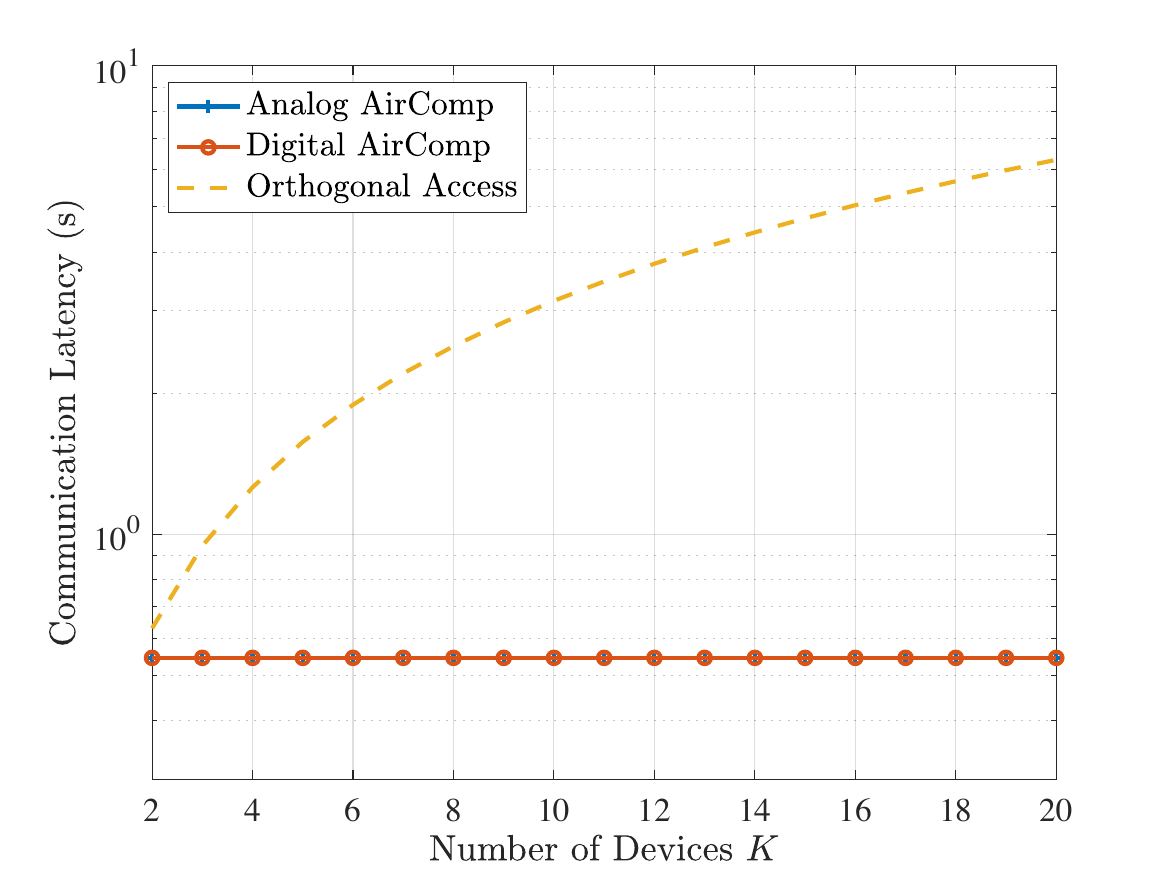}
    \caption{Comparison of latency among digital AirComp, analog AirComp and orthogonal access.}
    \label{fig:latency_vs_snr_benchmark}
\end{figure}

In Fig.~\ref{fig:latency_vs_snr_benchmark}, we compare the communication latency of digital AirComp against the benchmarking schemes. The latency is evaluated at $\gamma=\text{10 dB}$ for a varying number of devices from $\text{K = 2}$ to $\text{K = 20}$. 
The result shows that digital AirComp inherited the superior scalability of AirComp, as both AirComp schemes manage to keep the latency at a constant level, which does not scale with $K$. 
The performance gap between AirComp schemes and the orthogonal access scheme increases as $K$ grows, and reaches $10$-time at $20$ devices, demonstrating the advantage of digital AirComp as a low-latency wireless technique for distributed computing.

\begin{figure}
    \centering
    \includegraphics[width=\linewidth]{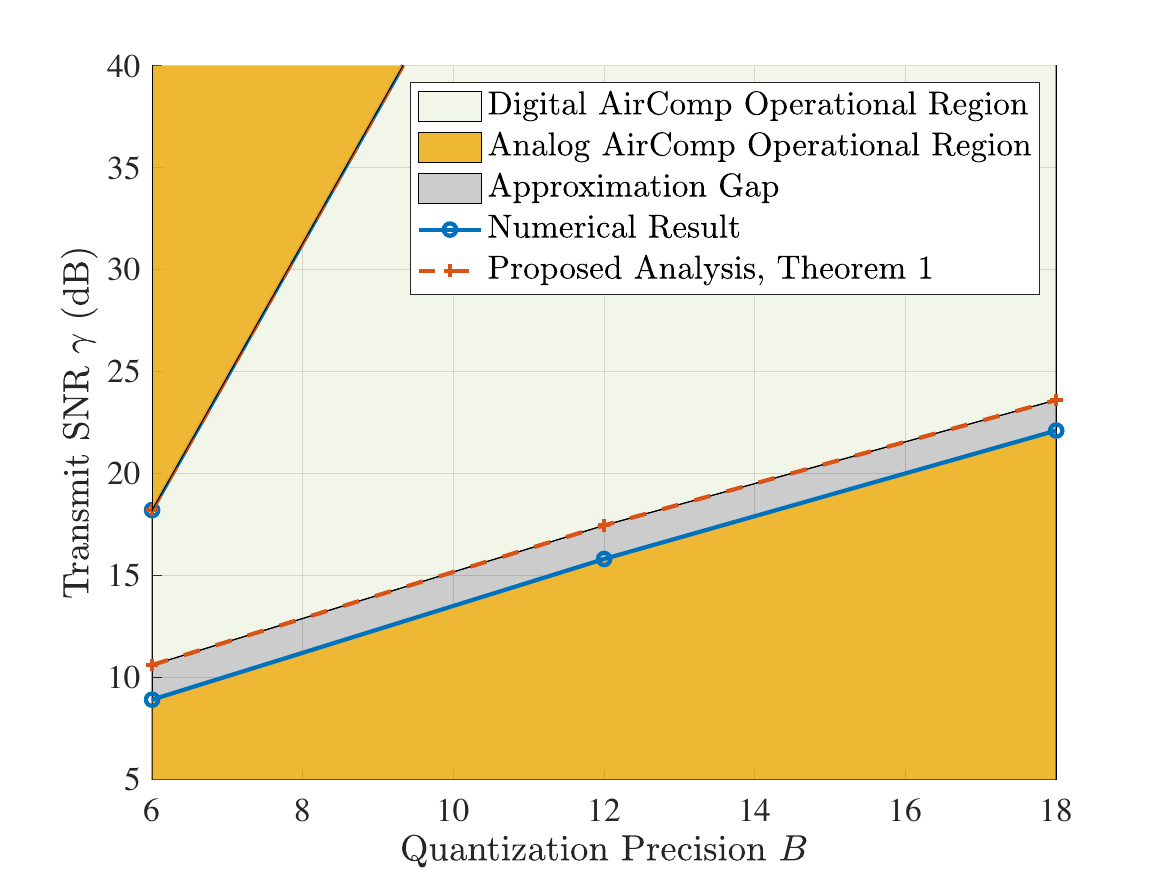}
    \caption{Operational region for digital AirComp with uniform bit-slicing and analog AirComp at different quantization precision.}
    \label{fig:snr_vs_bitWidth}
\end{figure}

Next, we compare the operational region of digital AirComp with uniform bit-slicing to analog AirComp at different quantization precisions using both numerical experiment and theoretical analysis presented in Theorem~\ref{theorem:digitalvsanalog}. 
The region is defined as the SNR range where the computation error of a given scheme is lower than that of the other scheme.
With $B = \{6,12,18\}$, and uniform bit-slicing, the respective curves are plotted in Fig.~\ref{fig:snr_vs_bitWidth}. 
It can be observed from Fig.~\ref{fig:snr_vs_bitWidth} that the operational region for digital AirComp covers the moderate to high SNR regime, while that for analog AirComp is complementary, spanning the low, and extremely high SNRs.
The operational region of digital AirComp increases with $B$, which is consistent with the theoretical analysis. 
The existence of a small approximation gap between the numerical and theoretical results can be noticed. This is because of the bounds and approximations used in the theoretical analysis.

\subsection{Performance of Adaptive Digital AirComp}
\subsubsection{Performance of Uniform-by-Best-Effort Bit-Slicing Scheme}
In Fig.~\ref{fig:adaptive_bitWidth}, we demonstrate the performance of the UBE bit-slicing scheme for various $B$. 
It can be observed from the figure that the saturation level of digital AirComp with UBE scheme decreases with $B$, which is due to the reduced quantization error. 
However, the curve slope responds to the change of $B$ in complicated ways. 
On one hand, the curve slope has no obvious change with the increment of $B$ when $B$ is relatively small (corresponding to $B=\{6,7,8\}$ in the figure). 
This is where the quantization error dominates the digital AirComp error, improvement in quantization precision significantly reduces the AirComp error. 
On the other hand, when $B$ is relatively large (corresponding to $B=\{8,9,10\}$ in the figure), the curve slope increases with $B$. This is because the additional bit-width is allocated to the more significant bits, which has more impact on the aggregation error, making them aggregation error dominated.

\begin{figure}
    \centering
    \includegraphics[width=\linewidth]{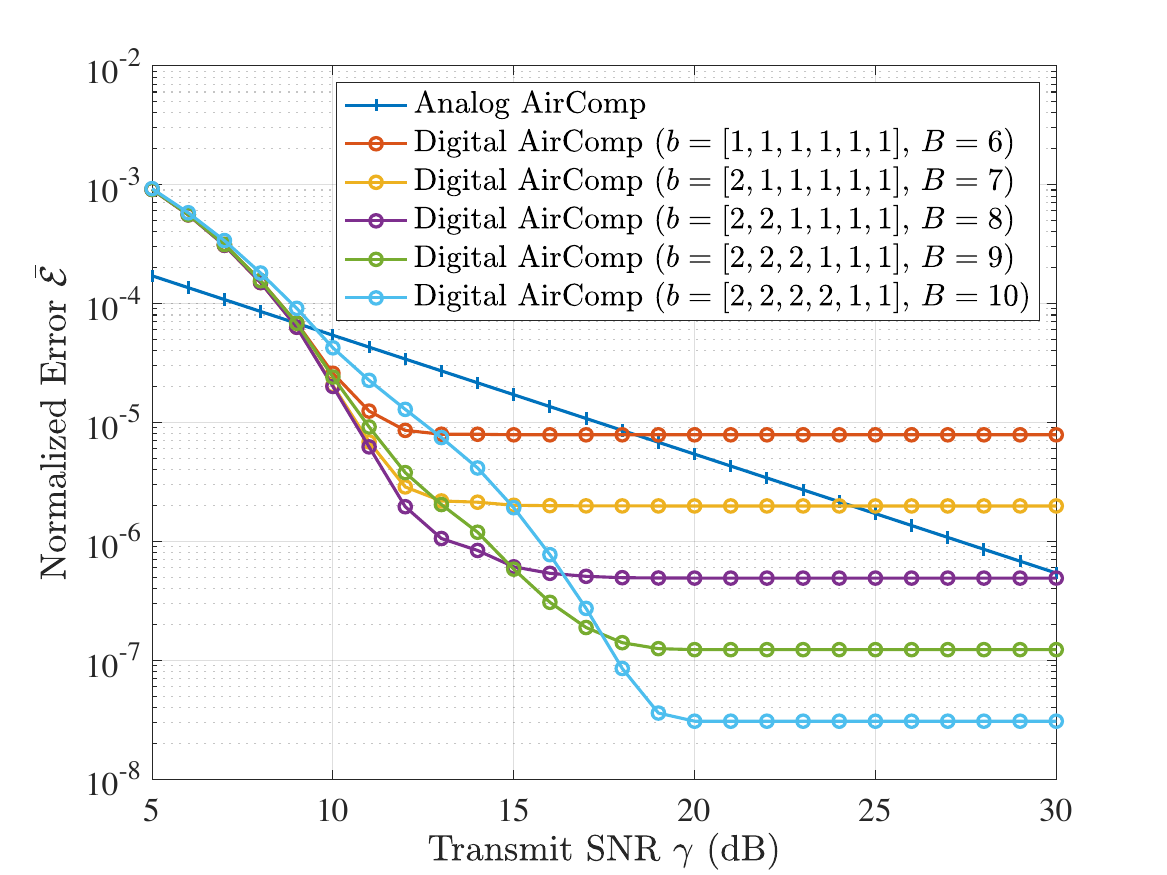}
    \caption{Normalized AirComp error of uniform-by-best-effort bit-slicing scheme.}
    \label{fig:adaptive_bitWidth}
\end{figure}

\subsubsection{Performance of Adaptive Digital AirComp}
Performance of the adaptive digital AirComp proposed in Algorithm~\eqref{alg:adaptive_digital_aircomp} is presented in Fig.~\ref{fig:optimal_bitSlicing_B10} in comparison with digital AirComp with UBE bit-slicing scheme and analog AirComp. 
Assuming $\bar{\gamma}=\text{18 dB}$, the optimal $B$ is found to be $B^\star=10$.
The optimal $\mathbf{b}$ are evaluated at integer $\gamma$ and labeled for several representative $\gamma$ values.
One can make following observations from Fig.~\ref{fig:optimal_bitSlicing_B10}: 
First, the bit-slicing scheme providing the lowest error at low SNR regime is the most unbalanced scheme (corresponding to $[5,1,1,1,1,1]$ in the figure) that prioritizes the protection of the important bits under the severe channel condition.
On the other hand, the best scheme at high SNR regime is $[2,2,2,1,1]$, which is exactly the UBE scheme introduced earlier. This scheme helps to improve the aggregation accuracy of all bits under favorable channel conditions.
As SNR increases, the optimal bit-slicing scheme gradually shifts from the unbalanced scheme to the UBE scheme.

\begin{figure}
    \centering
    \includegraphics[width=\linewidth]{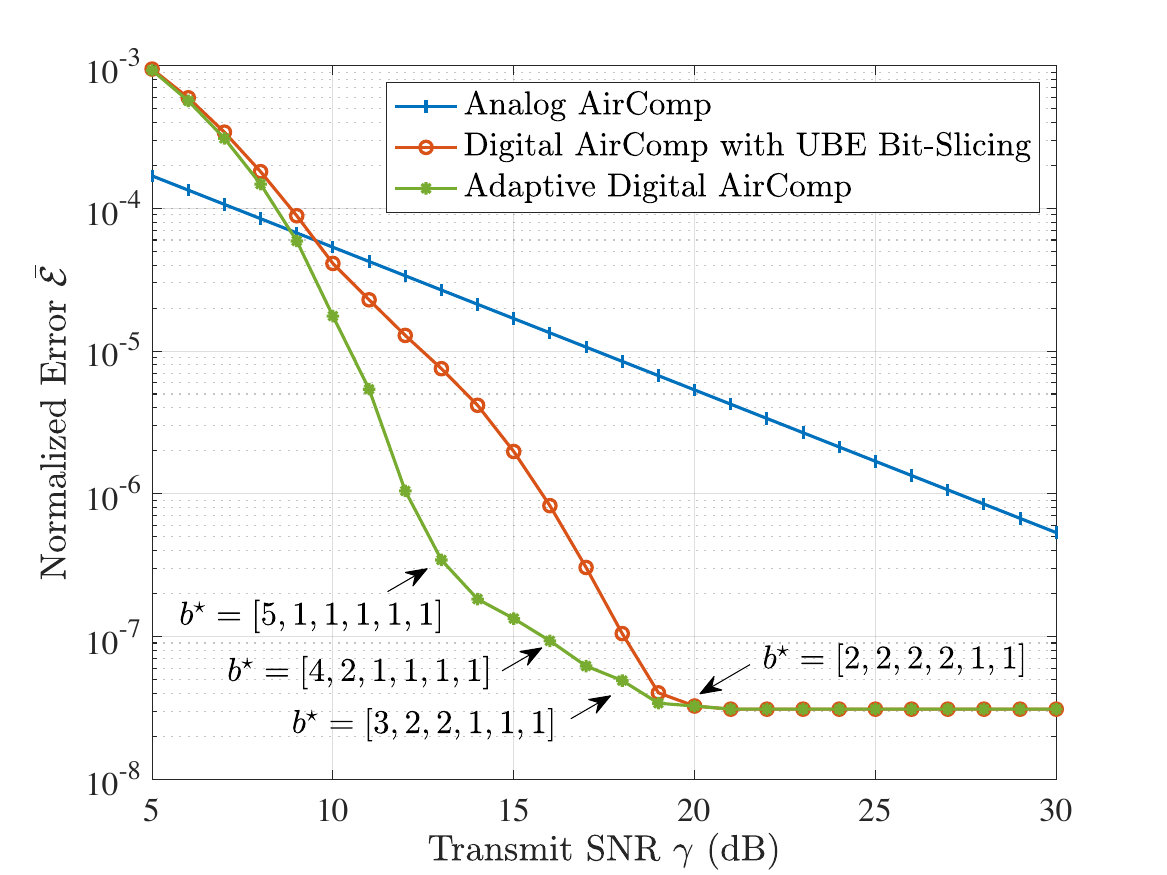}
    \caption{Performance comparison between adaptive digital AirComp, digital AirComp with uniform-by-best-effort bit-slicing and analog AirComp.}
    \label{fig:optimal_bitSlicing_B10}
\end{figure}

\section{Conclusion}
\label{sec:conclusion}
To simultaneously overcome the bottleneck in both communication latency and reliability for next-generation distributed edge networks, we proposed the framework of digital AirComp that realizes constant-time reliable aggregation of data generated on edge devices. 
The framework leverages the innate capability of QAM in translating superposition of symbols into magnitude addition of their represented values to overcome the latency bottleneck imposed by orthogonal access, and relies on optimized symbol detection to attain reliable aggregation results. 
Furthermore, a technique termed bit-slicing is applied to divide a data value with long bit width into multiple slices with short bit widths, resulting in not only the flexibility of adaptive modulation but also suppression of quantization error. Thereby, digital AirComp attains higher reliability than the traditional analog counterpart at practical SNRs.

The proposed digital AirComp framework in this work points to numerous directions for further investigations. 
For one, it is interesting to study the optimal control of the digital AirComp system by jointing exploiting the degrees-of-freedom including, for example, bit widths in the multi-type bit-slicing scheme, modulation orders and power allocation in time and frequency.
On the other hand, one benefit of digital AirComp is its compatibility with digital coding system. Mature channel coding techniques such as lattice coding can be redesigned to support coded digital AirComp to attain ever higher reliability.
Last, the flexibility of bit-level operations allows digital AirComp to seamlessly integrate with many distributed applications. For example, the importance-aware bit-slicing can be jointly optimized with feature pooling in distributed inference to realize one-shot aggregation and pooling over-the-air.

\appendix

\subsection{MAP Detection of QAM Constellation}
\label{app:QAMdetection}
The maximum a posteriori probability of superimposed QAM symbol is given by:
\begin{equation}
    \label{eq:maxaposteriori}
    \begin{aligned}
    \mathrm{max}\{p_j f(r|s_j)\} &= \mathrm{max}\{p^{(I)}_\ell p^{(Q)}_k f_r(r^{(I)}|s^{(I)}_\ell)f_r(r^{(Q)}|s^{(Q)}_k)\} \\
    &= \mathrm{max}\{p^{(I)}_\ell f_r(r^{(I)}|s^{(I)}_\ell)\}\\
    &\indent\cdot\mathrm{max}\{p^{(Q)}_k f_r(r^{(Q)}|s^{(Q)}_k)\},
    \end{aligned}
\end{equation}
where $\hat{s}^{(I)}$ and $\hat{s}^{(Q)}$ denotes the estimation of $s^{(I)}=\sum_{k=1}^K \Re(m_{k})$ and $s^{(Q)}=\sum_{k=1}^K \Im(m_{k})$, respectively, $r^{(I)}=\Re(r)$ and $r^{(Q)}=\Im(r)$ denotes the $I$ and $Q$ branch of the received signal, respectively, and $f_r(r|s)=\frac{1}{\sqrt{\pi\sigma_z^2}}\exp\left(-\frac{|r-\rho s|^2}{\sigma_z^2}\right)$ denotes the channel transition probability of the $I$ or $Q$ branch.

Note that the independence of the real and imaginary part of $s$ is a direct result of the independence of the real and imaginary part of the transmitted symbols $\{m_{k}\}$, which holds because the two sliced integers are generated from i.i.d data. Therefore, the first equality in \eqref{eq:maxaposteriori} holds. The second equality follows from the fact that $p^{(I)}_\ell f_r(r^{(I)}|s^{(I)}_\ell)\in(0,1)$ and $p^{(Q)}_k f_r(r^{(Q)}|s^{(Q)}_k)\in(0,1)$. As a result, the MAP detection of superimposed QAM symbol is equivalent to the separate MAP detection of the superimposed PAM symbols on its $I$ and $Q$ branches, i.e.,
\begin{equation}
    \hat{s}(r) = \hat{s}^{(I)}(r^{(I)})+j\hat{s}^{(Q)}(r^{(Q)}).
\end{equation}

\subsection{Proof of Lemma~\ref{lemma:exactdist}}
\label{app:exactdist}
Using the probability generating function (PGF) of $m$, which is defined by:
\begin{equation}
    \label{eq:pgfm}
    \begin{aligned} 
        \phi_{m}(t) &= \sum_{m_\ell\in\mathcal{P}}\Pr\{m=m_\ell\}t^{m_\ell}\\
        &= \frac{1}{P}\sum_{m_\ell\in\mathcal{P}}t^{m_\ell}.
    \end{aligned}
\end{equation}
Since $s$ is the sum of $K$ i.i.d random variables following the distribution of $m$, the PGF of $s$ is given by the $K$-th power of $\phi_{m}(t)$, i.e.,
\begin{equation}
    \label{eq:pgfs}
    \phi_{s}(t) = \frac{1}{P^K}\left(\sum_{m_\ell\in\mathcal{P}}t^{m_\ell}\right)^K.
\end{equation}
$\Pr\{s=s_n\}$ is the coefficient of $t^{s_n}$ in the expansion of \eqref{eq:pgfs}, coefficient of the polynomial term in \eqref{eq:pgfs} has been studied in\cite{steffen_eger_stirlings_2014}, and given by \eqref{eq:polycoef}.

\subsection{Proof of Lemma~\ref{lemma:aggregationerror}}
\label{app:detectionerror}
From \eqref{eq:detectionerror}, we study $\mathcal{E}_{det}$ from the asymptotic behavior of $P_{j\to m}(\sigma_z)$, as $\sigma_z$ approaches zero.
\begin{equation}
    \label{eq:symbolerrorlim}
    \lim_{\sigma_z \to 0} \frac{P_{j\to m}(\sigma_z)}{Q\left(\frac{\rho d}{\sqrt{2}\sigma_z}\right)}
    = \lim_{\sigma_z \to 0} \frac{Q\left(\frac{b_m^--\rho s_j}{\sigma_z/\sqrt{2}}\right)}{Q\left(\frac{\rho d}{\sqrt{2}\sigma_z}\right)}-\lim_{\sigma_z \to 0}\frac{Q\left(\frac{b_m^+-\rho s_j}{\sigma_z/\sqrt{2}}\right)}{Q\left(\frac{\rho d}{\sqrt{2}\sigma_z}\right)}.
\end{equation}
The first term in the right-hand side of \eqref{eq:symbolerrorlim} follows:
\begin{align}
    \lim_{\sigma_z \to 0} \frac{Q\left(\frac{b_m^--\rho s_j}{\sigma_z/\sqrt{2}}\right)}{Q\left(\frac{\rho d}{\sqrt{2}\sigma_z}\right)}
    &= \lim_{\sigma_z \to 0} \left(\frac{2\delta_{m,j}-d}{d}+\frac{\sigma_z^2\log\left(\frac{p_m}{p_{m-1}}\right)}{d^2\rho^2}\right) \nonumber \\
    &\indent \cdot
    \exp\left(-\frac{\delta_{m,j}(\delta_{m,j}-d)\rho^2}{\sigma_z^2}\right.\nonumber\\
    &\indent\indent+\frac{2\delta_{m,j}-d}{2d}\log\left(\frac{p_m}{p_{m-1}}\right)\\
    &\indent\indent\left.-\frac{\sigma_z^2\log^2\left(\frac{p_m}{p_{m-1}}\right)}{4d^2\rho^2}\right) \label{eq:symbolerrorlhopital} \\
    &\hspace*{-0.2\linewidth}= \begin{cases}
        \frac{2\delta_{m,j}-d}{d}\left(\frac{p_{m}}{p_{m-1}}\right)^{\frac{2\delta_{m,j}-d}{2d}}, & \text{if } (m-j)\in\{0,1\}\\
        0,              & \text{otherwise} 
    \end{cases} \label{eq:symbolerrorlim1}
\end{align}
where we omit the subscript $\ell$ from $d_\ell$ for brevity, $\delta_{m,j}=s_m-s_j=(m-j)d$ denotes the distance between $s_j$ and $s_m$, and \eqref{eq:symbolerrorlhopital} follows from the L'Hopital's rule. Similarly, the second term in the right-hand side of \eqref{eq:symbolerrorlim} follows:
\begin{align}
    \label{eq:symbolerrorlim2}
    \lim_{\sigma_z \to 0} \frac{Q\left(\frac{b_m^+-\rho s_j}{\sigma_z/\sqrt{2}}\right)}{Q\left(\frac{\rho d}{\sqrt{2}\sigma_z}\right)}&\nonumber\\
    &\hspace*{-0.2\linewidth}= \begin{cases}
        \frac{2\delta_{m,j}+d}{d}\left(\frac{p_{m+1}}{p_{m}}\right)^{\frac{2\delta_{m,j}+d}{2d}}, & \text{if } (m-j)\in\{-1,0\}\\
        0.              & \text{otherwise}
    \end{cases}
\end{align}
Substituting \eqref{eq:symbolerrorlim1} and \eqref{eq:symbolerrorlim2} into \eqref{eq:symbolerrorlim} and \eqref{eq:detectionerror} generates the asymptotic $\mathcal{E}_{det}$ as:

\begin{equation}
    \label{eq:detectionerror4}
    \lim_{\sigma_z \to 0} \mathcal{E}_{det}
    = d^2\sum_{s_j\in\mathcal{S}}\left(\sqrt{p_{j-1}p_j}+ \sqrt{p_j p_{j+1}} \right)Q\left(\frac{\rho d}{\sqrt{2}\sigma_z}\right)
\end{equation}
Substituting \eqref{eq:detectionerror4} to \eqref{eq:aggregationerror2} completes the proof.

\subsection{Proof of Theorem~\ref{theorem:digitalvsanalog}}
\label{app:digitalvsanalog}
By using the improved Chernoff bound of the Q-function, i.e., $Q(x)\leq\frac{1}{2}\exp\left(-\frac{x^2}{2}\right)$, to simplify the digital AirComp error, we derive the inequality of \eqref{eq:analogerror} and \eqref{eq:digitalerroruniform} as:
\begin{equation}
    \frac{AC}{2}\exp\left(-\frac{d^2\gamma}{4}\right)+\frac{K}{12}\leq\frac{4^B}{12L}\frac{1}{\gamma}.
\end{equation}
Rearranging the terms gives:
\begin{equation}
    \label{eq:digitalerrorlessthananalog}
    \gamma\exp\left(-\frac{d^2\gamma}{4}\right)+\frac{K}{6AC}\gamma\leq\frac{4^B}{6LAC},
\end{equation}
which matches the definition of $r$-Lambert $W$ function, that is, the solution of the equation:
\begin{equation}
    \label{eq:rlambertw}
    (x-t)\exp\left(cx\right)+r(x-t)=a
\end{equation}
is denoted by:
\begin{equation}
    x=t+\frac{1}{c}W_{re^{-ct}}(cae^{-ct}),
\end{equation}
with $t=0$, $c=-\frac{d}{4}$, $r=\frac{K}{6AC}$, and $a=\frac{4^B}{6LAC}$.

The $r$-Lambert function has three branches denoted by $W_{r,-2}$, $W_{r,-1}$ and $W_{r,0}$, respectively, when $r\in\left(0,1/e^2\right)$. The function value between $W_{r,-2}(cae^{-ct})$ and $W_{r,-1}(cae^{-ct})$ corresponds to \eqref{eq:digitalerrorlessthananalog}. This completes the proof.

\subsection{Proof of Corollary~\ref{corollary:limitofK}}
\label{app:limitofK}
The structure of $r$-Lambert function with $0<r<1/e^2$ is characterized as follows (and illustrated in Fig.~\ref{fig:lambertwithdiffr}):
\begin{itemize}
    \item Branch $W_{r,-2}$ is strictly increasing, and $W_{r,-1}$ is strictly decreasing, and the two branches are connected at point $(f_r(W_{-1}(-re)-1),W_{-1}(-re)-1)$.
    \item Branch $W_{r,0}$ is strictly increasing, and $W_{r,-1}$ and $W_{r,0}$ connects at point $(f_r(W_{0}(-re)-1),W_{0}(-re)-1)$.
\end{itemize}
Where $f_r(x)=xe^x+rx$. To guarantee both boundaries of \eqref{eq:snrregime} exist, the following condition is necessary:
\begin{equation}
    f_r(W_{0}(-re)-1) < -\frac{d^2}{4}\frac{4^B}{3LAC} < f_r(W_{-1}(-re)-1).
\end{equation}
While the first inequality holds for any $K>0$. The second inequality holds when:

\begin{align}
    -\frac{d^2}{4}\frac{4^B}{6LAC} & < -\frac{1}{8L} \label{eq:appEfirst} \\
    & < r\left[W_{-1}\left(-re\right)+\frac{1}{W_{-1}\left(-re\right)}-2\right] \\
    & = f_r(W_{-1}(-re)-1)
\end{align}
where \eqref{eq:appEfirst} follows from $C<\frac{4^B}{4^b-1}$ and $A<2$.

\bibliographystyle{IEEEtran}
\bibliography{references.bib}

\end{document}